\documentclass[a4paper,reqno,11pt]{amsart}
\usepackage{amsthm}
\usepackage{txfonts}
\usepackage{amssymb}
\usepackage{amsmath}
\usepackage{amsthm}
\usepackage[margin=1in]{geometry}
\usepackage{ esint }
\usepackage{enumitem}
\usepackage{physics}
\usepackage{tikz}
\usepackage{pgfplots}
\pgfplotsset{compat=newest}
\usepackage{dsfont}
\usepackage{amsaddr}
\numberwithin{equation}{section}

\newcommand{\subscript}[2]{$#1 _ #2$}

\newtheorem{thm}{Theorem}
\newtheorem{prop}[thm]{Proposition}
\newtheorem{lem}[thm]{Lemma}
\newtheorem{cor}[thm]{Corollary}

\newtheorem*{Rem}{Remark}

\theoremstyle{plain} 
\newcommand{\thistheoremname}{}
\newtheorem{genericthm}[thm]{\thistheoremname}

\begin{document}
\title{Localization and IDS Regularity in the Disordered Hubbard Model within Hartree-Fock Theory}
\author{Rodrigo Matos and Jeffrey Schenker}
\address{Department of Mathematics, Michigan State University, East Lansing MI 48823, USA}
\email{matosrod@msu.edu}
\maketitle

\date{November 25, 2016}

\begin{abstract}
{Using the fractional moment method it is shown  that, within the Hartree-Fock approximation for the disordered Hubbard Hamiltonian, weakly interacting Fermions at positive temperature exhibit localization, suitably defined as exponential decay of eigenfunction correlators. Our result holds in any dimension in the regime of large disorder and at any disorder in the one dimensional case. As a consequence of our methods, we are able to show H\"older continuity of the integrated density of states with respect to energy, disorder and interaction.}

\end{abstract}

\section{Introduction}
Our goal in this note is to study Anderson localization in the context of infinitely many particles. We shall formulate our results for the disordered Hubbard model within Hartree-Fock theory. However, as the techniques involved are quite flexible, we expect that similar statements can be made in a more general framework, under appropriate modifications of the decorrelation estimates on section \ref{improvesec}. The (deterministic) Hubbard model under Generalized Hartree-Fock Theory has been discussed (at zero and positive temperature) by Lieb, Bach and Solovej  in \cite{B-Lieb-S} but, to the best of our knowledge, the localization properties of the disordered version of this model remained unexplored, even in the context of restricted Hartree-Fock theory, up to the present work. The main difficulty lies on the addition of a self-consistent effective field, which will be random and non-local by nature, to a random Schr\"odinger operator.
 The conclusion of this note can be summarized as follows: under technical assumptions, the results on (single-particle) Anderson localization obtained in the non-interacting setting in the regimes of large disorder (in dimension $d\geq 2$) and at any disorder (in dimension $d=1$),  remain valid under the presence of sufficiently weak interactions. More specifically, in the regime of strong disorder this is accomplished in any dimension by theorem \ref{main} below. Theorem \ref{1dloc} contains the improvement in dimension one, where any disorder strength leads to localization, provided the interaction strength is taken sufficiently small.
 Our methods contain various bounds in the fluctuations of the effective interaction which are interesting on their own right and potentially useful on different contexts. To exemplify this, we prove H\"older regularity of the integrated density of states (IDS) with respect to various parameters by adapting arguments of \cite{H-K-S}, which is the content of theorem \ref{thmids}.\par
 \subsection{Discussion of the results and main obstacles}
   Mathematically, our setting can be understood as an Anderson-type model $H_{\omega}=H_0+\lambda U_{\omega}$ where the values of the random potential $U$ at different sites are correlated in a highly non-local and self-consistent fashion. The correlations are governed by a nonlinear function of $H_{\omega}$, as explained on section \ref{motivation}. In comparison to the recent result on Hartree-Fock theory for lattice fermions in \cite{Duc}, achieved via multiscale analysis, we use the fractional moment method to establish  exponential decay of the eigenfunction correlators at large disorder in any dimension but also at any disorder in dimension one. In particular, in the above regimes we obtain for any $t>0$,  exponential decay (on expectation) for the matrix elements of the Hamiltonian evolution, which means that, on average, $|\mel{m}{e^{-itH}}{n}|$ decays exponentially on $|m-n|$.\par
  The result of complete localization in dimension one in such interacting context is new and deserves attention on its own. Its main technical difficulty lies on the non-local correlations of the potential, which means that standard tools such as Furstenberg's theorem and Kotani theory are not available. Moreover, a large deviation theory for the Green's function is a further obstacle to establishing dynamical localization even if one obtains uniform positivity of the Lyapunov exponent. We overcome these challenges using ideas of \cite[Chapter 12]{A-W-B}, where arguments reminiscent of the proof of the main result in \cite{K-S} are presented. We then obtain positivity of the Lyapunov exponent at any disorder using uniform positivity for the Lyapunov exponent of the Anderson model, combined with an explicit bound on how this quantity depends on the interaction strength, see theorem \ref{uniformposh}. When it comes to establishing a large deviation theorem, our modification of the argument in \cite[Theorem 12.8]{A-W-B} relies on quantifying the decorrelations on the effective potential, which is presented on lemma \ref{mixinglem} in the form of a strong mixing statement.
 It is worth clarifying that, since our proof is based on fractional moments, we have not established localization in one dimension for rough potentials as in  \cite{C-K-M}. Moreover, the gap assumption in \cite{Duc} is replaced by working at positive temperature thus our results do not apply to Hartree-Fock ground states.
 \subsection{Hartree-Fock theory}
 Hartree-Fock theory has been widely applied in computational physics and chemistry. It also has a rich mathematical literature which goes well beyond the scope of random operators, see for instance \cite{H-Lewin-S},\cite{G-H-Lewin},\cite{B-Lieb-S},\cite{B-L-L-S},\cite{Lieb-Simon} and references therein.
 \subsection{Background on Localization for interacting systems}
 The main results of this note lie in between the vast literature on (non-interacting) single particle localization and the recent efforts to study many particle systems, as in the case of an arbitrary, but finite, number of particles in the series of works by Chulaevsky-Suhov \cite{C-S1},\cite{C-S2},\cite{C-S3} and Aizenman-Warzel \cite{A-W-P}. In comparison to the later, we only seek for a single-particle localization result but allow for infinitely many interactions, which occur in the form of a mean field. 
 In comparison to the recent developments on spin chains, as the study of the XY spin chain in \cite{H-S-St} and the droplet spectrum of the XXZ quantum spin chain in \cite{E-K-St} and \cite{B-W}, the notions of localization for a single-particle effective Hamiltonian are more clear and can be displayed from pure point spectrum to exponential decay of eigenfunctions and exponential decay of eigenfunction correlators. The later is agreed to be the strongest form of single particle localization and it is what we accomplish in this manuscript. If fact, dynamical localization in the form of theorems \ref{1dloc} and \ref{main} implies pure point spectrum via the RAGE theorem (see \cite[Proposition 5.3]{Stolz}) and exponential decay of eigenfunctions, see \cite[Theorem 7.2 and Theorem 7.4]{A-W-B}.

 \section{Definitions and Statement of the Main Result}
\subsection{Notation}In what follows, $\mathbb{Z}^d$ will be equipped with the norm
$|n|=|n_1|+...+|n_d|$ for $n=(n_1,...,n_d)$. Given a subset 
$\Lambda\subset \mathbb{Z}^d$, we define 
$\ell^2(\Lambda):=\{\varphi:\Lambda \rightarrow \mathbb{C}\,| 
\sum_{n\in \Lambda}|\varphi(n)|^2<\infty\}$ and, for 
$\varphi\in \ell^2(\Lambda)$, we let $\|\varphi\|_{\ell^2(\Lambda)}:=\left(\sum_{n\in \Lambda}|\varphi(n)|^2<\infty\}\right)^{1/2}$. Throughout this note, $\eta$ will be a positive constant and $F_{\beta,\kappa}$ will denote the Fermi-Dirac function at inverse temperature $\beta>0$ and chemical potential $\kappa$: \begin{equation}\label{fermidirac}
F_{\beta,\kappa}(z)=\frac{1}{1+e^{\beta(z-\kappa)}}.
\end{equation}
We shall omit the dependence on the above parameters whenever it is clear from the context.
For many of our bounds, the specific form of (\ref{fermidirac}) is not important and $F$ could denote a fixed function which is analytic on the strip $\mathcal{S}=\{z\in \mathbb{C}:\,\, |\mathrm{Im}z|<\eta\}$ and continuous up to the boundary of $\mathcal{S}$, in which case we define $\|F\|_{\infty}:=\sup_{z\in \mathcal{S}}|F(z)|$. For the function $F_{\beta,\kappa}$ in (\ref{fermidirac}) one can take $\eta=\frac{\pi}{2\beta}$. However, to obtain robust results which incorporate delicate fluctuations, further properties of the Fermi-Dirac function are necessary. Namely, in section \ref{improvesec} we use the the fact that $tF(t)$ is bounded as $t\to \infty$ and that $t(1-F(t))$ is bounded  $t\to -\infty$. These properties will also play a role in the decoupling estimates needed in the proof of theorem \ref{1dloc} but could be relaxed if one is only interested in the large disorder proof of theorem \ref{main} for a specific distribution with heavy tails (for instance, the Cauchy distribution).\par
 Our main goal is to study localization properties of non-local perturbations of the Anderson model $H_{\mathrm{And}}:=-\Delta+\lambda V_{\omega}$ which naturally arise in the context of Hartree-Fock theory for the Hubbard model. The random potential $V_{\omega}$ is the multiplication operator on $\ell^2(\mathbb{Z}^d)$  defined as
 \begin{equation}\label{potentialdef}\left(V_{\omega}\varphi\right)(n)=\omega_n\varphi(n)
\end{equation} for all $n\in \mathbb{Z}^d$ and $\{\omega_n\}_{n\in \mathbb{Z}^d}$ are independent, identically distributed random variables on which we impose technical assumptions described in the next paragraph. The hopping operator $\Delta$ is the discrete Laplacian on $\mathbb{Z}^d$, defined via \begin{equation}\left(\Delta\varphi\right)(n)=\sum_{|m-n|=1}\left(\varphi(m)-\varphi(n)\right).
\end{equation} The proofs of localization via fractional moments usually do not require the hopping to be dictated by $\Delta$; below we will replace $\Delta$ by a more general operator $H_0$ whose matrix elements decay sufficiently fast away from the diagonal. It is technically useful to formulate some of our results in finite volume, i.e, we will work with restrictions of the operators to $\ell^2(\Lambda)$ but the estimates obtained will be volume independent, meaning that all the constants involved are independent of $\Lambda\subset \mathbb{Z}^d$. We will use $\mathds{1}_{\Lambda}$ to denote the characteristic function of $\Lambda$ as well as the natural projection $P_{\Lambda}:\ell^2(\mathbb{Z}^d) \rightarrow \ell^2(\Lambda)$.
  With these preliminaries we are ready to define the Schr\"odinger operators studied in this work.
\subsection{Definition of the operators}
Let  $H_{\mathrm{And}}=H_0+\lambda V_{\omega}$ be the Anderson model on $\ell^2\left(\mathbb{Z}^d\right)$ where:
\begin{enumerate}[label=(\subscript{A}{{\arabic*}})]
\item{ $$\zeta(\nu):=\sup_{m}\sum_{n\in \mathbb{Z}^d}|H_0(m,n)|\left(e^{\nu|m-n|}-1\right)<\frac{\eta}{2}, \,\,\,\,\mathrm{for\, some}\,\,\nu>0 \,\,\mathrm{fixed}.$$}

\item{ $V_{\omega}$ is defined as in $(\ref{potentialdef})$ and the random variables $\{\omega(n)\}_{n\in \mathbb{Z}^d}$ are independent, identically distributed with a density $\rho$: 
$$\mathbb{P}\left(\omega(0)\in I\right)=\int_{I} \rho(x)\,dx,\,\,\,\,\,\mathrm{for}\,\,\, I\subset \mathbb{R}\,\,\,\mathrm{a\,\, Borel\,\, set\,}.$$}
\item{ We also assume that $\mathrm{supp}\,\rho=\mathbb{R}$ with 
\begin{equation}\label{flucassump}\frac{\rho(x)}{\rho(x')}\geq e^{-c_1(\rho)|x-x'|(1+c_2(\rho)\max\{\,|x|,|x'|\,\})}
\end{equation}
for some $c_1(\rho)>0$ and $c_2(\rho)\geq 0$ and any $x,x'\in \mathbb{R}$.}
\end{enumerate}
Before stating the remaining assumptions on $\rho$, we need to introduce some notation. Assume that $\rho$ satisfies (\ref{flucassump}). Let
\begin{equation}\label{upperdensity}
\overline{\rho}(x)=\frac{\rho(x)}{\int^{\infty}_{-\infty}\rho(\alpha)h(x-\alpha)\,d\alpha}
\end{equation}
where
\begin{equation}\label{fluctintegral3}h(x)=
\left\{
\begin{array}{lll}
e^{-{\overline{c}}_{\rho}|x|} \;\,\,\,\,\, \mathrm{if} \;c_2(\rho)=0. \\
 \\
e^{-{\overline{c}}_{\rho}|x|^2}\,\,\,\,\, \; \mathrm{if} \;c_2(\rho)>0. \;\;\;
\end{array}
\right.
\end{equation}

\begin{enumerate}[resume,label=(\subscript{A}{{\arabic*}})]
\item\label{fluctintegral1}{ The function $\overline{\rho}$ is bounded for some ${\overline{c}}_{\rho}>0$.
}
\end{enumerate}
\begin{Rem}\label{remark2}
The technical assumptions $(A_3)-(A_4)$ will be needed for the large disorder result of theorem \ref{main} below. They include, for instance, the Cauchy distribution, the Gaussian,  and the exponential distribution $\rho(v)=\frac{m}{2}e^{-m|v|}$.
\end{Rem}

The above assumptions will suffice to show localization at large disorder on theorem \ref{main} below. To show complete localization in dimension one, theorem \ref{1dloc} will also require a moment condition on $\overline{\rho}$, which is the following. 
\begin{enumerate}[resume,label=(\subscript{A}{{\arabic*}})]

\item\label{momentassumption}{ For some $\varepsilon>0$ and some ${\overline{c}}_{\rho}>0$, $\int^{\infty}_{-\infty}|x|^{\varepsilon}\overline{\rho}(x)\,dx<\infty$.}
\end{enumerate}

\begin{Rem}\label{remark3}
The assumption $(A_5)$ covers, for example, the Gaussian and the exponential distributions but it does not cover the Cauchy or other distribution with heavy tails. It will be necessary for the one dimensional result of theorem \ref{1dloc} below. More specifically, this requirement will imply a moment condition which will be needed to relate the Green's function to the Lyapunov exponent, see sections \ref{1dideas} and \ref{detailslowerbound}.
\end{Rem}
\begin{Rem}\label{remark1} The specific bound on $\zeta(\nu)$ is necessary to ensure that the Combes-Thomas bound 
$|G(m,n;E+i\eta)|\leq \frac{2}{\eta}e^{-\nu|m-n|}$ holds \cite[Theorem 10.5]{A-W-B}, where $G(m,n;z)$ denotes $\mel{m}{(H-z)^{-1}}{n}$, whenever this quantity is defined. \end{Rem}

 \par

Define the operator $H_{\mathrm{Hub}}$, acting on $\ell^{2}\left(\mathbb{Z}^d\right)\oplus \ell^{2}\left(\mathbb{Z}^d\right)$  by \begin{equation}\label{Hubbarddefmain}H_{\mathrm{Hub}}=\begin{pmatrix}
\,H_{\uparrow}(\omega) & 0 \\
0 & H_{\downarrow}(\omega)\,\\
\end{pmatrix}
:=
\begin{pmatrix}
\,H_0+\lambda V_{\omega}+gV_{\uparrow}(\omega) & 0 \\
0 & H_0+\lambda V_{\omega}+gV_{\downarrow}(\omega)\,\\
\end{pmatrix}
\end{equation}
where the operators $H_{\uparrow}(\omega)$ and $H_{\downarrow}(\omega)$ act on $\ell^{2}\left(\mathbb{Z}^d\right)$ and the so-called effective potentials are defined via
\begin{equation}\label{eff}
\begin{pmatrix}
\,V_{\uparrow}(\omega)(n) \\
V_{\downarrow}(\omega)(n)\,\\
\end{pmatrix}
=\begin{pmatrix}
\,\mel{n}{F(H_{\downarrow})}{n} \\
\mel{n}{F(H_{\uparrow})}{n}\,\\
\end{pmatrix}.
\end{equation}
Note that the above equations only define $H_{\uparrow}(\omega)$ and $H_{\downarrow}(\omega)$ implicitly. Existence and uniqueness of 
$V_{\uparrow}$ and $V_{\downarrow}$ will be shown in section \ref{existHubb} via a fixed point argument. The model (\ref{Hubbarddefmain}) is usually referred to as the Hartree approximation, due to the absence of exchange terms. In section \ref{motivation} below we will show that the terminology Hartree-Fock approximation is justified when $g<0$, which represents a repulsive interaction.

The Hubbard model is schematically represented in the following picture. The black (horizontal) edges represent hopping between sites and the red (vertical) edges represent the effective interaction between the two layers, which are identical copies of $\mathbb{Z}^d$.
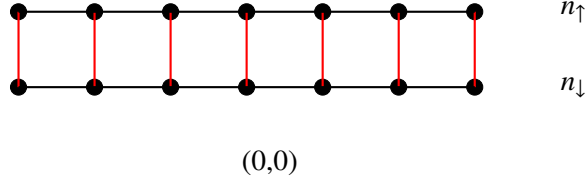
\vspace{1cm}
\begin{center}
\begin{tikzpicture}[scale=1.0]
\draw[fill=black] (-1,1) circle (3pt);
\draw[fill=black] (-1,2) circle (3pt);

\draw[fill=black] (0,1) circle (3pt);
\draw[fill=black] (0,2) circle (3pt);

\draw[fill=black] (1,1) circle (3pt);
\draw[fill=black] (1,2) circle (3pt);

\draw[fill=black] (2,1) circle (3pt);
\draw[fill=black] (2,2) circle (3pt);

\draw[fill=black] (3,1) circle (3pt);
\draw[fill=black] (3,2) circle (3pt);

\draw[fill=black] (4,1) circle (3pt);
\draw[fill=black] (4,2) circle (3pt);

\draw[fill=black] (5,1) circle (3pt);
\draw[fill=black] (5,2) circle (3pt);

\node at (2.5,0) {(0,0)\,\,\,\,\,\,};
\node at (6.5,1) {$n_{\downarrow}$\,\,\,\,\,\,};
\node at (6.5,2) {$n_{\uparrow}$\,\,\,\,\,\,};
\
\draw[red, thick]
(5,1)--(5,2);
\draw[red, thick]
(4,1)--(4,2);
\draw[red, thick]
(3,1)--(3,2);
\draw[red, thick]
(2,1)--(2,2);
\draw[red, thick]
(1,1)--(1,2);
\draw[red, thick]
(0,1)--(0,2);
\draw[red, thick]
(-1,1)--(-1,2);
 \draw[thick]
 (-1,2)--(0,2) --(1,2) -- (2,2) -- (3,2) --(4,2) -- (5,2);
 \draw[thick]
 (-1,1)--(0,1) --(1,1) -- (2,1) -- (3,1) --(4,1) -- (5,1);

\end{tikzpicture}
\end{center}

\subsection{Main Theorems}
 Fix an interval $I\subset \mathbb{R}$ and  define the eigenfunction correlator through 
\begin{equation} Q_{I}(m,n):=\sup_{|\varphi|\leq 1}\left(|\mel{m}{\varphi(H_{\uparrow})}{n}|+|\mel{m}{\varphi(H_{\downarrow})}{n}|\right).\end{equation}
The operators $H_{\uparrow}$ and $ H_{\downarrow}$ and defined as in (\ref{Hubbarddefmain}) and the supremum being taken over Borel measurable functions bounded by one and supported on the interval $I$. In case $I=\mathbb{R}$ we simply write $Q(m,n)$.
Our first result is the following:
\begin{thm}\label{1dloc} In dimension $d=1$, let $H_0=-\Delta$ and assume that the conditions $(A_1)-(A_5)$ hold. For any 
$\lambda>0$ and any closed interval $I\subset \mathbb{R}$ , there is a constant $g_1>0$ such that whenever $|g|<g_1$ we have 
\begin{equation}\label{dynloc1d}
 \mathbb{E}\left(Q_I(m,n)\right)\leq Ce^{-\mu_{1}|m-n|}.
 \end{equation} for any  $m,n\in \mathbb{Z}^d$ and positive constants $\mu_1=\mu_1(\lambda,\nu,\eta,I)$, $C(\eta,g,\lambda,\|F\|_{\infty},I)$.

\end{thm}

\begin{thm}\label{main}
Suppose that the conditions $(A_1)-(A_4)$ hold. For any dimension $d\geq 1$, there exists a constant $g_d=g(d,\eta,\|F\|_{\infty},\nu)$ such that, whenever $|g|<g_d$, there is a positive constant $\lambda_0(g)$ for which
\begin{equation}\label{dynloc}
 \mathbb{E}\left(Q(m,n)\right)\leq Ce^{-\mu_{d}|m-n|}.
 \end{equation}
holds for $\lambda >\lambda_0(g)$, any $m,n\in \mathbb{Z}^d$ and some positive constants $\mu_{d}=\mu(d,\lambda,g,\nu,\eta)$, $C(\eta,\nu,d,g,\lambda,\|F\|_{\infty})$.
\end{thm}
\begin{Rem} It will follow from the proof that the constant $g_d$ in theorem \ref{main} can be taken proportional to $\frac{\eta\left(1-e^{-\nu}\right)^d}{\|F\|_{\infty}}$.
\end{Rem}
\begin{Rem} The constant $g_1$ in theorem \ref{1dloc} can be taken equal to be the minimum among a factor proportional to $\frac{\eta\left(1-e^{-\nu}\right)}{\|F\|_{\infty}}$ and the upper bound obtained in corollary \ref{uniformpos}, which also depends on the lower bound for the Lyapunov exponent of the Anderson model on $\ell^2\left(\mathbb{Z}\right)$.
\end{Rem}

Recall the definition of the integrated density of states for an ergodic operator $H$ :
\begin{equation}N_H(E)=\lim_{|\Lambda|\to \infty}\frac{\mathrm{Tr}P_{(-\infty,E)}(\mathds{1}_{\Lambda}H\mathds{1}_{\Lambda})}{|\Lambda|}.\end{equation}
 For the definition of ergodic operator one may consult \cite[Definition 3.4]{A-W-B}.
In what follows, we denote by $N_0(E)$ the corresponding quantity for the free operator $H_0$ defined above, which is assumed to be ergodic for the result below, where we shall be concerned with the small disorder regime and aim for bounds which do not depend upon $\lambda$ as $\lambda \to 0$.

\begin{thm}\label{thmids}
Assume that $(A_1)-(A_2)$ hold with $x^2\rho(x)$ bounded and that $g^2<\lambda$ . Fix a interval $I$ where $E\mapsto N_0(E)$ is $\alpha_0$-H\"older continuous and a bounded interval $J\subset \mathbb{R}$.
The integrated density of states $N_{\lambda,g}(E)$ of $H_{\mathrm{Hub}}$ is H\"older continuous with respect to $E$ and with respect to the pair $(\lambda,g)$. More precisely:

\begin{enumerate}[label=(\subscript{IDS}{{\arabic*}})]
\item{
For $E,E' \in I$ 
\begin{equation}\label{IDSenergy}
 |N_{\lambda,g}(E)-N_{\lambda,g}(E')|\leq C(\alpha,I,g)|E-E'|^{\alpha}
\end{equation} for $\alpha\in [0,\frac{\alpha_0}{2+\alpha_0}]$ and $C(\alpha,I,g)$ independent of $\lambda$.
}
\item{\label{IDSdisorder} If $\lambda,\lambda'\in J$, we have that, for any $E\in I$, $\alpha\in [0,\frac{\alpha_0}{2+\alpha_0}]$ and $\beta\in [0,\frac{2}{\alpha+3d+4}]$,
\begin{equation}|N_{\lambda,g}(E)-N_{\lambda',g'}(E)|\leq C(\alpha_0.d,I)\left(|\lambda-\lambda'|^{\beta}+|g-g'|^{\beta}\right).
\end{equation}}

\end{enumerate}
\end{thm}

\section{Motivation}\label{motivation}
We shall explain the motivation for the above choice of the effective potential. We are only going to outline the derivation of the self-consistent equations as this is a standard topic, see, for instance, \cite[Chapter 3]{Kurig}.\par Let $\Lambda \subset \mathbb{Z}^d$ be a finite set. Following the notation of 
\cite{B-Lieb-S}, we use $\Gamma$ to denote a one particle density matrix, i.e, a $2\times 2$ matrix whose entries are operators on $\ell^2\left(\Lambda\right)$ and which satisfies $0\leq \Gamma \leq \mathds{1}$. We then write
$$\Gamma
:=
\begin{pmatrix}
\,\Gamma_{\uparrow} & \Gamma_{\uparrow \downarrow} \\
\Gamma_{\downarrow \uparrow} &\Gamma_{\uparrow}\,\\
\end{pmatrix}
$$
where $\Gamma_{\downarrow \uparrow}=\Gamma^{\dag}_{\uparrow \downarrow}$.

  As in \cite[Equation 3a.8]{B-Lieb-S}, the pressure functional $\mathcal{P}(\Gamma)$ is defined as 
 \begin{equation}-\mathcal{P}(\Gamma)=\mathcal{E}(\Gamma)-\beta^{-1}\mathcal{S}(\Gamma).
 \end{equation}
 The energy functional is
 \begin{equation}\label{energy}
 \mathcal{E}(\Gamma)=
 \mathrm{Tr}\left( H_0-\kappa+\lambda V_{\omega}\right)\Gamma
 +g\sum_{n}\mel{n}{\Gamma_{\uparrow}}{n}\mel{n}{\Gamma_{\downarrow}}{n},
\end{equation}
where we have identified $H_0-\kappa+\lambda V_{\omega}$ with
$ \begin{pmatrix}
\,H_0-\kappa+\lambda V_{\omega} &0 \\
0 &H_0-\kappa+\lambda V_{\omega}\,\\
\end{pmatrix}$.\par The entropy is given by
\begin{equation}\mathcal{S}(\Gamma)=-\mathrm{Tr}\left(\Gamma \log \Gamma+(1-\Gamma)\log (1-\Gamma)\right)
.\end{equation}
Generally, the choice of energy functional (\ref{energy}) is referred to as Hartree approximation as exchange terms are neglected. However, in the case of a repulsive interaction among the particles, it is easy to prove that such exchange terms do not affect the choice of minimizer for $-\mathcal{P}(\Gamma)$ and the process may be referred to as the Hartree-Fock approximation. Indeed, the Hartree-Fock energy for the repulsive interaction would incorporate the term $-g|\mel{n}{\Gamma _{\uparrow \downarrow}}{n}|^2$, which is non-negative when $g<0$. Thus, for repusive interactions, off-diagonal terms can be disregarded for minimization purposes, see the analogue discussion in \cite[Section 4a]{B-Lieb-S}.
The minimizer $\Gamma$ of $-\mathcal{P}(\Gamma)$ exists since $\Lambda$ is a finite set. Moreover, it satisfies
\begin{equation}\label{densityHubbard}
\mel{n}{\Gamma_{\uparrow}}{n}=\mel{n}{\frac{1}{1+e^{\beta(H_0-\kappa+\lambda V_{\omega}+\mathrm{Diag}\left(\Gamma_{\downarrow})\right)}}}{n}.
\end{equation}
\begin{equation}
\mel{n}{\Gamma_{\downarrow}}{n}=\mel{n}{\frac{1}{1+e^{\beta(H_0-\kappa+\lambda V_{\omega}+\mathrm{Diag}\left(\Gamma_{\uparrow})\right)}}}{n}.
\end{equation}

Thus, the effective Hamiltonian on $\ell^2\left(\Lambda\right)\oplus \ell^2\left(\Lambda\right)$ is determined by

$$H^{\Lambda}_{\omega}
:=
\begin{pmatrix}
\,H_0+\lambda \omega(n)+gV^{\Lambda}_{\uparrow}(n) & 0 \\
0 & H_0+\lambda \omega(n)+gV^{\Lambda}_{\downarrow}(n)\,\\
\end{pmatrix}
$$

\begin{equation}
\label{potentialHubbard1}
V^{\Lambda}_{\uparrow}(\omega)(n):=\mel{n}{\frac{1}{1+e^{\beta(H_0-\kappa+\lambda \omega+gV_{\downarrow})}}}{n}
\end{equation}
\begin{equation}
\label{potentialHubbard2}
V^{\Lambda}_{\downarrow}(\omega)(n):=\mel{n}{\frac{1}{1+e^{\beta(H_0-\kappa+\lambda \omega+gV_{\uparrow})}}}{n}.
\end{equation}

It will follow from arguments given below that if $\Lambda_R$ is an increasing sequence with 
$\cup_{R\in \mathbb{N}}\Lambda_R=\mathbb{Z}^d$ then, for fixed $m \in \mathbb{Z}^d$, \begin{equation}
\lim_{R\to \infty} V^{\Lambda_R}_{\mathrm{eff}}(m)=V_{\mathrm{eff}}(m)
\end{equation} and this fact ensures that, for localization purposes in the Hubbard model, it suffices to study $H_{\mathrm{Hub}}$ and its finite volume restrictions.

\section{ Outline of the Proof of theorem \ref{main}}
We now want to outline the proof of the theorem \ref{main} in the related model where $H_{\mathrm{Hub}}$ is replaced by the operator
\begin{equation}\label{toymodel}
H=H_0+\lambda\omega(n)+gV_{\mathrm{eff}}(n)
\end{equation} acting on $\ell^2\left(\mathbb{Z}^d\right)$ with
\begin{equation}
\label{potentialHubbard2}
V_{\mathrm{eff}}(n)=\mel{n}{\frac{1}{1+e^{\beta(H_0+\lambda \omega+gV_{\mathrm{eff}})}}}{n}.
\end{equation}
In this case, the correlator is defined as
\begin{equation} Q_{I}(m,n):=\sup_{|\varphi|\leq 1}|\mel{m}{\varphi(H)}{n}|.\end{equation}
where $\varphi$ is Borel measurable and supported on $I$.\par
The above operator exhibits the main mathematical features of the Hubbard model, namely: the effective potential is defined self-consistently as a non-local and non-linear function of $H$. Thus, it is natural to first illustrate our methods here. For now let's assume the existence and uniqueness of $V_{\mathrm{eff}}$ are proven as well as its regularity with respect to $\{\omega(n)\}_{n\in \mathbb{Z}^d}$. Combined with estimates on the derivatives of $V_{\mathrm{eff}}$, the above facts form a significant portion of the proof which is developed in sections \ref{existencesection} and \ref{regularitysection}. The, somewhat straightforward, extension of the proof to $H_{\mathrm{Hub}}$ will be explained in section \ref{hubbardext}. 
A feature which theorem \ref{1dloc} and theorem \ref{main} have in common is that the eigenfunction correlator decay will be achieved via the Green's function of $H^{\Lambda}=\mathds{1}_{\Lambda}H\mathds{1}_{\Lambda}$, which is $H$ restricted to finite sets $\Lambda \subset \mathbb{Z}^d$. Let  \begin{equation}\label{greensdef}
G^{\Lambda}(m,n,z)=\mel{m}{(H^{\Lambda}-z)^{-1}}{n}.
\end{equation}
Using the basics of the fractional moment method, which dates back to \cite{A-M} and \cite{Aiz}, we aim at showing that, for some $s\in(0,1)$,
\begin{equation}\label{Greendecay}\mathbb{E}\left(\Big|G^{\Lambda}(m,n;z)\Big|^s\right)\leq Ce^{-\mu_d|m-n|}
\end{equation}
holds uniformly in $z\in \mathbb{C}^{+}$, with positive constants $C=C(d,s,g,\lambda,\nu,\eta,\|F\|_{\infty})$ and $\mu(d,s,g,\lambda,\nu,\eta,\|F\|_{\infty})$ independent of the volume  $|\Lambda|$. In this context, the Green's function decay expressed by equation (\ref{Greendecay}) implies
\begin{equation}\label{eigencorr}
\mathbb{E}\left(Q(m,n)\right)\leq C'e^{-\mu_{d}'|m-n|}
\end{equation}
 for some exponent $\mu'_d=\mu'(d,s,g,\lambda,\nu,\eta,\|F\|_{\infty})>0$ and $C'=C'(\eta,\nu,d,g,\lambda,s,\|F\|_{\infty})$. This is well known and explained in great generality, for instance, in \cite[Theorem A.1]{A-S-F-H}.\par
 Another aspect which is shared by the proofs of theorems \ref{1dloc} and \ref{main} is that the starting point  to obtain (\ref{Greendecay}) will be the following \emph{a-priori} bound.
 \begin{lem}\label{apriori} Given a finite set $\Lambda \subset \mathbb{Z}^d$, there exist a constant $C_{\mathrm{AP}}=C_{\mathrm{AP}}(\eta,\nu,d,g,\lambda,s,\|F\|_{\infty})$, independent of $\Lambda$, such that
 \begin{equation}\mathbb{E}\left(\Big|G^{\Lambda}(m,n;z)\Big|^s\right)\leq C_{\mathrm{AP}}
\end{equation}
holds for any $m,n\in \Lambda$.
\end{lem}
The proof of lemma \ref{apriori} will follow from lemma \ref{bdddensity} below. Let
\begin{equation}\label{U}
 U_{\omega}(n)=\omega(n)+\frac{g}{\lambda}V_{\mathrm{eff}}(n,\omega).
 \end{equation}
 be the ``full'' potential at site $n$.
 From now on, to keep the notation simple, we drop the dependence on $\omega$ in the new variables $\{U(n)\}_{n\in \Lambda}$. Note that $U(n)$ and $U(m)$ are correlated for all values of $m$ and $n$. The strategy is to show that, for $g$ sufficiently small, they still behave as if they were independent in the following sense:

\begin{lem}\label{bdddensity}
Fix $\Lambda \subset \mathbb{Z}^d$ finite and $n_0\in \Lambda$. The conditional distribution of $U(n_0)=u$ at specified values of $\{U(n)\}_{n\in {\Lambda \setminus\{n_0\}}}$  has density ${\rho}^{\Lambda}_{n_0}$. Moreover, under assumptions $(A_1)-(A_4)$ we have that
\begin{equation}\sup_{\Lambda}\sup_{n_0\in \Lambda} \sup_{u\in \mathbb{R}}{\rho}^{\Lambda}_{n_0}(u)<\infty.
\end{equation}
If, additionally, assumption \ref{momentassumption} holds then ${\rho}^{\Lambda}_{n_0}(u)\in L^{1}\left(\mathbb{R},|x|^{\varepsilon}dx\right)$.
 \end{lem}
 
The proof of the above result is detailed in section \ref{proofoflemma}; it requires exponential decay of $|\frac{\partial V_{\mathrm{eff}}(n)}{\partial \omega(m)}|$ and $|\frac{\partial^2 V_{\mathrm{eff}}(n)}{\partial \omega(m)\omega(l)}|$ with respect to $|m-n|$ and $|m-n|+|l-n|$, respectively. The need for this decay is the main reason to require $\beta>0$ or, in other words, to require analiticity of $F$ on a strip. The intuitive explanation for lemma \ref{bdddensity} is that the random variables $U(n)$ and $U(n_0)$ decorrelate in a strong fashion as $|n-n_0|$ becomes large. 
Lemma \ref{bdddensity} implies (\ref{Greendecay}) for any $0<s<1$ as long as  $\lambda$ is taken sufficiently large,
see \cite[Theorem 10.2]{A-W-B}. 

The proof of theorem \ref{1dloc} will require additional efforts involving tools which are specific to one dimension, which we shall comment on below.
\section{One dimensional aspects: strategy of the proof of theorem \ref{1dloc}}\label{1dideas}

The argument for proving theorem \ref{1dloc} follows closely the approach in the proof of theorem 12.11 in \cite{A-W-B}, which we now recall.
\subsection{Main ideas in the i.i.d case}
In the reference \cite[Chapter 12]{A-W-B}  Green's function decay is described in terms of the moment generating function, defined by
  \begin{equation}\label{momentgen}
 \varphi(s,z) =\lim_{|n|\to \infty}\frac{\ln \mathbb{E}{\left(|G(0,n;z)|^s\right)}}{|n|}.
 \end{equation}
 The existence of the above quantity for all $z\in \mathbb{C}^{+}$ and $s\in (0,1)$ and its relationship to the Lyapunov exponent are a consequence of Fekete's lemma:
 \begin{lem}[Fekete]{\label{Feketeclass}}
Let $\{a_n\}_{n\in \mathbb{N}}$ be a sequence of real numbers such that, for every pair $(m,n)$ of natural numbers,
\begin{equation}\label{subadc}
a_{n+m}\leq a_n+a_m
\end{equation}
Then, $\alpha=\lim_{n\to \infty} \frac{a_n}{n}$
exists and equals $\inf_{n\in \mathbb{N}}\frac{a_n}{n}$.
\end{lem}
It is an elementary observation that if, instead, the sequence $\{a_n\}_{n\in \mathbb{N}}$ satisfies $a_{n+m}\leq a_n+a_m+C$ then the above result applies to $b_n:=a_n+C$ and that an analogous statement holds for superadditive sequences, which satisfy (\ref{subadc}) with the inequality reversed.
In the i.i.d. context, the sequence
$a_n=\ln \mathbb{E}\left(|G(0,n;z)|^s\right)$ is shown to be both subbaditive and superadditive, meaning that there exist constants $C_{-}(s,z)$ and $C_{+}(s,z)$ for which
\begin{equation}\label{upperandloweradd}
a_n+a_m+C_{-} \leq a_{n+m}\leq a_n+a_m+C_{+}.
\end{equation}
holds for all $m,n \in \mathbb{N}$, see \cite[Lemma 12.10]{A-W-B}.
A consequence of this fact, together with a precise control of the arising constants, is stated below.
\begin{lem}\label{greenmoment} \cite[Theorem 12.8]{A-W-B}
 For any $z\in \mathbb{C}^{+}$, there are 
 $c_s(z),C_s(z)\in (0,\infty)$ such that for all $n\in \mathbb{Z}$
 \begin{equation}
c^{-1}_{s}(z)e^{\varphi(s,z)|n|}\leq \mathbb{E}\left(|G_{\mathrm{And}}(0,n;z)|^s\right)\leq C_s(z)e^{\varphi(s,z)|n|}.
\end{equation} Moreover, for any compact set $K\subset \mathbb{R}$ and $S\subset [-1,1)$, we have the local uniform bound \begin{equation}
\sup_{s\in S}\sup_{z\in K+i(0,1]}\max\{c_s(z),C_s(z)\}<\infty
\end{equation} 
and the same result holds with $z$ replaced by its boundary value $E+i0$ for Lebesgue almost every $E$.
\end{lem}
On the other hand, for fixed $z\in \mathbb{C}^{+}$,  $\varphi(s,z)$ is shown to be convex function of $s$ and non-increasing in $[-1,+\infty)$, with its derivative at $s=0$ satisfying $\frac{\partial \varphi(0,z)}{\partial s}=-\mathcal{L}(z)$. It is a consequence of these facts that for almost every $E\in \mathbb{R}$ there exists a value $s=s(E)\in (0,1)$ such that
 \begin{equation}\label{boundmomentgeniid}
 \varphi(s,E)\leq -\frac{s}{2}\mathcal{L}(E).
 \end{equation}
 The above is the content of \cite[Equation (12.86)]{A-W-B}. Dynamical localization is shown to hold locally as a consequence of the inequality (\ref{upperandloweradd}) along with  lemma \ref{greenmoment}, the inequality (\ref{boundmomentgeniid}) and Kotani theory, which establishes that $\mathcal{L}(E)$ is positive for almost every $E\in \mathbb{R}$.
\subsection{Modifications}
In this section we will outline the proof theorem \ref{1dloc} with $H_{\mathrm{Hub}}$ again replaced by the operator $H$ on $\ell^2\left(\mathbb{Z}^d\right)$ defined in (\ref{toymodel}). For simplicity we set $\lambda=1$ since the disorder strenght does not play an important role in theorem \ref{1dloc}. Let $H_{+}=H_{[0,\infty)\cap\mathbb{Z}}$ be the restriction of $H$ to $\ell^2\left(\mathbb{Z}^{+}\right)$ and denote by $G^{+}(m,n;z)$ the Green's function of $H^{+}$. Recall the definition of the Lyapunov exponent: initially, for $z\in \mathbb{C}^{+}$, we let \begin{equation}
\mathcal{L}(z)=-\mathbb{E}\left(\ln|G^{+}(0,0;z)|\right).
\end{equation} By Herglotz theory (see, for instance, \cite[Appendix B]{A-W-B} and references therein) it is seen that, for Lebesgue almost every $E\in \mathbb{R}$, $\mathcal{L}(E)$ is well defined as $\lim_{\delta \to 0^{+}} \mathcal{L}(E+i\delta)$. Finally, recall the uniform positivity of the Lyapunov exponent for the Anderson model on $\ell^2\left( \mathbb{Z}\right)$:
 \begin{equation}
 \mathrm{ess}\inf_{E\in \mathbb{R}}\mathcal{L}_{\mathrm{And}}(E)>\mathcal{L}_{\mathrm{And}}
 \end{equation}
 for some $\mathcal{L}_{\mathrm{And}}>0$. The first step towards Green's function decay (\ref{Greendecay}) will be showing uniform positivity of $\mathcal{L}(E)$, which is accomplished by the following.
\begin{thm}\label{uniformposh} There exists a constant $C_{\mathrm{Lyap}}(s,\eta,g,\|F\|_{\infty})>0$ such that
\begin{equation}
|\mathcal{L}(z)-\mathcal{L}_{\mathrm{And}}(z)|\leq C_{\mathrm{Lyap}}|g|^s
\end{equation} for all $z\in \mathbb{C}^{+}$.

\end{thm}
\begin{proof}

From the resolvent identity we obtain
\begin{equation}
\frac{|G^{+}(0,0;z)|}{|G^{+}_{\mathrm{And}}(0,0;z)|}\leq 1+|g|\|F\|_{\infty}\sum_{n}|G^{+}(0,n;z)|\frac{|G^{+}_{\mathrm{And}}(n,0;z)|}{|G^{+}_{\mathrm{And}}(0,0;z)|}
\end{equation}

\begin{equation}
\frac{|G^{+}_{\mathrm{And}}(0,0;z)|}{|G^{+}(0,0;z)|}\leq 1+|g|\|F\|_{\infty}\sum_{n}|G^{+}_{\mathrm{And}}(0,n;z)|\frac{|G^{+}(n,0;z)|}{|G^{+}(0,0;z)|}
\end{equation} 
Using the bound $\ln(1+x)\leq \frac{x^{s}}{s}$ for $0<s<1$ and $x>0$ we reach, for $0<s<1/2$,
\begin{equation}\ln\left(\frac{|G^{+}(0,0;z)|}{|G^{+}_{\mathrm{And}}(0,0;z)|}\right)\leq \frac{|g|^{s}}{s}\|F\|^{s}_{\infty}\sum_{n}|G^{+}(0,n;z)|^s\frac{|G^{+}_{\mathrm{And}}(n,0;z)|^s}{|G^{+}_{\mathrm{And}}(0,0;z)|^s}.
\end{equation}
Taking expectations, using the definition of the Lyapunov exponents and the Cauchy-Schwarz inequality 
\begin{equation}\mathcal{L}_{\mathrm{And}}(z)-\mathcal{L}(z)\leq \frac{|g|^{s}}{s}\|F\|^{s}_{\infty} \sup_n \mathbb{E}\left(|G^{+}(0,n;z)|^{2s}\right)^{1/2}\sum_n\mathbb{E}\left(\frac{|G^{+}_{\mathrm{And}}(n,0;z)|^{2s}}{|G^{+}_{\mathrm{And}}(0,0;z)|^{2s}}\right)^{1/2}:=C_{\mathrm{Lyap}}(s,\eta,\nu,\|F\|_{\infty})|g|^{s}.
\end{equation}
The fact that $C_{\mathrm{Lyap}}$ is a finite quantity follows from a couple of remarks. Firstly, by Feenberg's expansion \cite[Theorem 6.2]{A-W-B} we have the identity
\begin{equation}\label{feenberg}
|G^{+}_{\mathrm{And}}(n,0;z)|=|G^{+}_{\mathrm{And}}(0,0;z)||G^{+}_{\mathrm{And}}(1,n;z)|
\end{equation}
where $G^{+}_{\mathrm{And}}(1,n;z)$ denotes the Green's function of $H_{\mathrm{And}}$ restricted to $\ell^{2}\left(\mathbb{Z} \right)\cap[1,\infty)$. From the 
\emph{a-priori} fractional moment bound on lemma  \ref{apriori} combined with the Green's function decay for dimensional Anderson model
\begin{equation}\mathbb{E}\left(|G^{+}_{\mathrm{And}}(1,n;z)|^{2s}\right)<C(s)e^{-\mu_{\mathrm{And}}|n|}
\end{equation}
we conclude that that $C_{\mathrm{Lyap}}<\infty$.
The estimate for $\mathcal{L}(z)-\mathcal{L}_{\mathrm{And}}(z)$ is similar.
\end{proof}

In principle one might worry that the pre-factor $C_{\mathrm{Lyap}}$ on the above bound will depend on $g$. However, it is easy to see from the arguments in the proof of lemma \ref{apriori}, that $C_{\mathrm{AP}}$ converges to a finite quantity as $g\to 0$, thus we shall disregard its dependence on $g$.

\begin{cor}\label{uniformpos} Whenever $|g|<\left(\frac{\mathcal{L}_{\mathrm{And}}}{C_{\mathrm{Lyap}}}\right)^{1/s}$ holds for some $s\in (0,1/2)$, we have 
\begin{equation}
\mathcal{L}_{0}:=\mathrm{ess}\inf_{E\in \mathbb{R}}\mathcal{L}(E)>0.
\end{equation}
\end{cor}
We can now proceed to the second step of the proof of theorem \ref{1dloc}, which consists of establishing Green's function decay from corollary \ref{uniformpos}. For that purpose, an important detail to keep in mind is that, in the correlated context, if we choose $a_n=\log\mathbb{E}\left(|G(0,n;z)|^s\right)$, the condition (\ref{upperandloweradd}) will not be fulfilled for all pairs $(m,n)$ due to the lack of independence between the potentials. This means that  Fekete's lemma is not applicable. Moreover, its well-studied modifications (for instance by P. Erd\"os and N. G. de Bruijn  \cite{Erdos}) do not seem to suffice either.

 To the best of our knowledge the result given below is new. Its formulation takes into account the strong decorrelation between the potentials in the Hubbard model and introduces a notion of approximate subbaditivity.
\begin{lem}[Fekete-type lemma for approximately subbaditive sequences]{\label{Fekete}}

Let $\delta>0$ be given and $\{a_n\}_{n\in \mathbb{N}}$ be a sequence of real numbers such that, for every triplet $m, n, r$ of natural numbers with $r\geq \delta\max\{\log m,\log n\}$, the inequality
\begin{equation}\label{subad}
a_{n+m+r}\leq a_n+a_m+C
\end{equation}
holds with a constant $C$ independent of $m,n$ and $r$.
 Then,
\begin{equation}
\alpha=\lim_{n\to \infty} \frac{a_n}{n} 
\end{equation} 
exists and equals $\inf_{n\in \mathbb{N}}\frac{a_n+C}{n}$ . Moreover, $\alpha\in [-\infty,0]$.
\end{lem}
Note that, as a consequence, we have
\begin{equation}\label{relatemomgren}
a_n\geq n\alpha -C
\end{equation}
for all $n\in \mathbb{N}$, where $C$ is the same constant as in (\ref{subad}).
The following decoupling estimate guarantees the applicability of the above lemma with the choice $a_n=\log\mathbb{E}_{[0,n]}\left(|\hat G(0,n;z)|^s \right)$, where $\hat G(0,n;z)=\mel{0}{(H_{[0,n]}-z)^{-1}}{n}$ is the Green's function of the operator $H$ restricted to $\ell^2\left([0,n]\cap \mathbb{Z}\right)$ and $\mathbb{E}_{[0,n]}$ denotes the expectation with respect to $U(0),...U(n)$.
\begin{lem}\label{mixinglem}[Strong mixing decoupling]
There exist positive numbers $C_{\mathrm{Dec}}(s,\nu,\eta,g,\|F\|_{\infty})$ and $\delta=\delta(\eta,\nu,g,\|F\|_{\infty})$ such that the inequality
\begin{equation}
\mathbb{E}_{[0,n+m+r]}\left(|\hat G(0,n+m+r;z)|^s \right)\leq  C_{\mathrm{Dec}}
\mathbb{E}_{[0,n]}\left(|\hat G(0,n;z)|^s \right) \mathbb{E}_{[0,m]}\left(|\hat G(0,m;z)|^s \right)
\end{equation}
holds whenever $r\geq \delta \log\max\{m,n\}$.
\end{lem}
A combination of lemmas \ref{mixinglem}, \ref{Fekete} and equation (\ref{relatemomgren}) yields the lower bound
\begin{equation}\label{lowerbound}
C^{-1}_{\mathrm{Dec}}e^{\varphi(s,z)n}\leq \mathbb{E}_{[0,n]}\left(|\hat G(0,n;z)|^s \right)\,\,\,\,\,\mathrm{for\,\, all}\,\,\,n\in \mathbb{N}.
\end{equation}
As we shall see in section \ref{detailslowerbound} below,
an application of the lower bound (\ref{lowerbound}) in combination with the superadditive version of lemma \ref{Fekete} applied to the sequence $b_n=-\varphi(s,E)n+\log\mathbb{E}_{[0,n]}\left(|\hat G(0,n;z)|^s \right)$ is enough to establish an upper bound
\begin{equation}\label{upperbound}
\mathbb{E}\left(|\hat G(0,n;z)|^s \right)\leq C(s,z)e^{\varphi(s,z)n}\,\,\,\,\,\mathrm{for\,\, all}\,\,\,n\in \mathbb{N}.
\end{equation}
where the constant $C(s)$ is locally uniform in $(s,z)\in (0,1)\times \mathbb{C}^{+}$.

After obtaining an analogue of lemma \ref{greenmoment} , the final step will be to relate the moment-generating function to the Lyapunov exponent through an inequality of the type
 \begin{equation}\label{boundmomentgen}
 \varphi(s,E)\leq -\frac{s}{2}\mathcal{L}_{0}.
 \end{equation} In reference \cite{A-W-B}, the bound (\ref{boundmomentgen}) is stated with $\mathcal{L}_{0}$ replaced by $\mathcal{L}(E)$ and with $s$ depending on $E$. However, it is easy to see from the arguments given there that $s$ can be chosen locally uniformly in $E$, see \cite[Equations (12.79) and (12.80)]{A-W-B}. Moreover, by making use uniform positivity of the Lyapunov exponent obtained in corollary \ref{uniformpos} we reach the inequality (\ref{boundmomentgen}).
The Green's function decay follows from the bounds (\ref{upperbound}) and (\ref{boundmomentgen}).\par The remainder of the paper is organized as follows. Sections \ref{existencesection}, \ref{regularitysection} and \ref{decaysection} establish existence, regularity and decay properties of the effective potential in the model (\ref{toymodel}). Section \ref{proofoflemma} contains the proof of lemma \ref{bdddensity}. Section \ref{hubbardext} explains the required modifications for the Hubbard model. The proof of theorem \ref{main} is finished combining sections \ref{proofoflemma} and \ref{hubbardext}.  Section \ref{details1d} covers the remaining details of the proof of theorem \ref{1dloc} and section \ref{idssection} contains the proof of theorem \ref{thmids}.

\section{Existence of the effective potential}\label{existencesection}
To justify the definition of the effective potential in (\ref{toymodel}), let $\Phi(V):\ell^{\infty}(\mathbb{Z}^d)\rightarrow \ell^{\infty}(\mathbb{Z}^d)$ be given by
 $\Phi(V)(n):=\mel{n}{F(T+\lambda V_{\omega}+gV)}{n}.$
 Recall that $F$ is analytic, bounded on the strip $S=\{|\mathrm{Im} z|<\eta\}$ and continuous up to the boundary of $S$. 
 Our goal is to check that $\Phi$ is a contraction in $\ell^{\infty}\left( \mathbb{Z}^d\right)$, meaning that 
 \begin{equation}\label{contractiondef}\|\Phi(V)-\Phi(W)\|_{\ell^{\infty}(\mathbb{Z}^d)}<c\|V-W\|_{\ell^{\infty}(\mathbb{Z}^d)}
 \end{equation}
  holds for some $c<1$ and all $V,W \in \ell^{\infty}(\mathbb{Z}^d)$.
  Using the analiticity of $F$ we have the following representation   \cite[Equation (D.2)]{A-Graf}
  
  \begin{equation} F(T+\lambda V_{\omega}+gV)=\frac{1}{2\pi i}\int^{\infty}_{-\infty}\left(\frac{1}{T+\lambda V_{\omega}+gV-i\eta+t}-\frac{1}{T+\lambda V_{\omega}+gV+i\eta+t}\right) f(t)\,dt
  \end{equation}
   for all $V\in  \ell^{\infty}(\mathbb{Z}^d)$, where $ f=F_{+} + F_{-} + D\ast F $ for $F_{\pm}(u)=F(u\pm i\eta)$ and $D(u)=\frac{\eta}{\pi\left( \eta^2+u^2\right)}$ is the Poisson kernel. It follows immediately that $\|f\|_{\infty}\leq 3\|F\|_{\infty}$. This is a prelude for the following fixed point argument, where the operator $T$ will be assumed to satisfy
   \begin{equation}\label{offdecayssump}
   \sup_{n}\sum_m|T(m,n)|\left(e^{\nu|m-n|}-1\right)<\frac{\eta}{2}.
   \end{equation}

\begin{prop}\label{Contraction} 
\begin{enumerate}[label=(\subscript{C}{\arabic*})]
\item{  For any self-adjoint operator $T$ on $\ell^2\left(\mathbb{Z}^d\right)$ satisfying  (\ref{offdecayssump}) and bounded potentials $V,W$, we have, for any $\nu'\in (0,\nu)$, that
 \begin{equation}\label{contraction1}  \Big|\mel{m}{\left(F(T+V)-F(T+W)\right)}{n}\Big|\leq \frac{72\sqrt{2}e^{-\nu'|m-n|}}{\eta\left(1-e^{\nu'-\nu}\right)^d}\|F\|_{\infty}\|V-W\|_{\infty}.
\end{equation}}
\item{ For any self-adjoint operator $T$ on $\ell^2\left(\mathbb{Z}^d\right)\oplus \ell^2\left(\mathbb{Z}^d\right)$ satisfying (\ref{offdecayssump}) and bounded potentials $V,W$ on $\ell^2\left(\mathbb{Z}^d\right)\oplus \ell^2\left(\mathbb{Z}^d\right)$ we have, for any $\nu'\in (0,\nu)$, that \begin{equation}\label{contraction2}  \Big|\mel{m}{\left(F(T+V)-F(T+W)\right)}{n}\Big|\leq \frac{144\sqrt{2}e^{-\nu'|m-n|}}{\eta\left(1-e^{\nu'-\nu}\right)^{d}}\|F\|_{\infty}\|V-W\|_{\infty}
\end{equation}
}

\item{ For any $m,n,j \in\mathbb{Z}^d$, the matrix elements $\mel{m}{F(T+gV)}{n}$ are differentiable with respect to $V(j)$ and \begin{equation}\label{contraction3}\Big|\frac{\partial\mel{m}{F(T+gV)}{n}}{\partial V(j)}\Big|\leq |g|\frac{72\sqrt{2}e^{-\nu\left(|m-j|+|n-j|\right)}}{\eta}\|F\|_{\infty}\|V\|_{\infty}.
\end{equation}}

\end{enumerate}
\end{prop}
\begin{proof} The resolvent identity gives
\begin{align*}
& \mel{m}{\frac{1}{T+V-t-i\eta}-\frac{1}{T+W-t-i\eta}}{n}+\mel{m} {\frac{1}{T+W-t+i\eta}-\frac{1}{T+V-t+i\eta}}{n}\\
&=\mel{m}{(\frac{1}{T+V-t-i\eta}-\frac{1}{T+V-t+i\eta})(W-V)\frac{1}{T+W-t-i\eta}}{n}-\\
&\mel{m}{ \left(\frac{1}{T+W-t+i\eta}-\frac{1}{T+W-t-i\eta}\right)(W-V)\frac{1}{T+V-t+i\eta}}{n}.\\
\end{align*}
 Taking absolute values in the first term on the right-hand side we obtain

\begin{align*}&\Big|\mel{m}{\left(\frac{1}{T+V-t-i\eta}-\frac{1}{T+V-t+i\eta}\right)(W-V)\frac{1}{T+W-t-i\eta}}{n}\Big| \\
&\leq \sum_{l\in\mathbb{Z}^d}|G^{V}(m,l;t+i\eta)-G^{V}(m,l;t-i\eta)||(W-V)(l)|G^{W}(l,n;t+i\eta)|\\
&\leq 24\sum_{l}|(V-W)(l)|e^{-\nu\left(|l-n|+|l-m|\right)} \mel{m}{\frac{1}{(T+V-t)^2+{\eta^2}/2}}{m}^{1/2}
\mel{l}{\frac{1}{(T+V-E)^2+{\eta^2}{/2}}}{l}^{1/2}.\\
\end{align*}
In the last step we made use of the Combes-Thomas bound $|G^{W}(m,n;t+i\eta)|\leq \frac{2}{\eta}e^{-\nu|m-n|}$ as well as lemma 3 in \cite[appendix D]{A-Graf} to estimate the difference between the Green functions  as
\begin{eqnarray}
|G^{V}(m,l;t+i\eta)-G^{V}(m,l;t-i\eta)|\leq 12\eta e^{-\nu|m-l|} \mel{m}{\frac{1}{(T+V-t)^2+{\eta^2}/2}}{m}^{1/2}
\mel{l}{\frac{1}{(T+V-E)^2+{\eta^2}{/2}}}{l}^{1/2}.
\end{eqnarray}
Integrating over $t$ we conclude, using Cauchy-Schwarz and the spectral measure representation, that
\begin{eqnarray}\label{resolventdiff}&\int^{\infty}_{-\infty}\left|\mel{m}{ \left(\frac{1}{T+V-t-i\eta}-\frac{1}{T+V-t+i\eta}\right)(W-V)\frac{1}{T+W-E-i\eta}}{n}\right|\,dt\\
&\leq
\frac{24\sqrt{2}\pi}{\eta}\sum_{l}|(V-W)(l)|e^{-\nu\left(|l-n|+|l-m|\right)}.
\end{eqnarray}
The above implies that
\begin{align*}
 &\frac{1}{2\pi}\int^{\infty}_{-\infty}\Big| \mel{m}{ \left(\frac{1}{T+V-t-i\eta}-\frac{1}{T+V-t+i\eta}\right)(W-V)\frac{1}{T+W-t-i\eta}}{n}\Big|\,dt\\
 &\leq \frac{12\sqrt{2}}{\eta}\|V-W\|_{\infty}e^{-\nu'|m-n|}\sum_{l\in \mathbb{Z}^d}e^{(\nu'-\nu)|l-n|}\\
 &=\frac{12\sqrt{2}}{\eta}\|V-W\|_{\infty}e^{-\nu'|m-n|}\frac{1}{\left(1-e^{\nu'-\nu}\right)^d}.\\
 \end{align*}
As a similar bound holds for
$\frac{1}{2\pi}\int^{\infty}_{-\infty}\Big|\mel{m}{\left(\frac{1}{H_0+W-t+i\eta}-\frac{1}{H_0+W-t-i\eta}\right)(V-W)\frac{1}{H_0+V-t+i\eta}}{n}\Big|\,dt$
, we conclude the proof of the inequality (\ref{contraction1}) by recalling that $\|f\|_{\infty}\leq 3 \|F\|_{\infty}$. The inequality (\ref{contraction2}) in the statement of proposition \ref{Contraction} follows from the same argument with the only difference that one has to sum two geometric series, hence the modification on the upper bound. The bound (\ref{contraction3}) is proven similarly: note that
 $h\frac{12\sqrt{2}\pi}{\eta}e^{-\nu|j-n|} e^{-\nu|m-j|}$ is an upper bound for the left-hand side of equation (\ref{resolventdiff}) with $V$ replaced by $gV$ and $W=g(V+hP_{j})$, where $P_j$ denotes the projection onto $\mathrm{Span}\{\delta_j\}$. We also observe that this time there will be no summation over $l$, hence the introduction of the $\nu'$ is unnecessary. We then conclude that
 \begin{equation}\label{differencequotient}
 \Big|\frac{\mel{m}{F(T+gV+hP_j)}{n}-\mel{m}{F(T+gV)}{n}}{h}\Big|\leq \frac{72\sqrt{2}\pi}{\eta}e^{-\nu|j-n|} e^{-\nu|m-j|}.
  \end{equation}
Letting $h\to 0$ finishes the proof.
\end{proof}
Taking $m=n$, as a consequence of the above proposition, (\ref{contractiondef}) holds whenever \begin{equation}\label{smallg}
|g|<\frac{\eta\left(1-e^{-\nu}\right)^d }{72\sqrt{2}\|F\|_{\infty}}.
\end{equation}
This observation yields the following.

\begin{prop}\label{existencepot}
Let $g_d=\frac{\eta\left(1-e^{-\nu}\right)^d}{72\sqrt{2}\|F\|_{\infty}}$. Then, for $|g|<g_d$, there is a unique effective potential $V_{\mathrm{eff}}\in \ell^{\infty}\left(\mathbb{Z}^d\right)$ satisfying 
\begin{equation}\label{toy}
V_{\mathrm{eff}}(n)=\mel{n}{F(H_0+\lambda \omega+gV_{\mathrm{eff}}}{n}.
\end{equation}
Moreover, for $\Lambda \subset \mathbb{Z}^d$, there is a unique 
$V^{\Lambda}_{\mathrm{eff}}$ in $ \ell^2\left(\Lambda\right)$ satisfying (\ref{toy}) with $H$ replaced by $H^{\Lambda}=\mathds{1}_{\Lambda}H\mathds{1}_{\Lambda}$.
\end{prop}
\begin{Rem} Replacing $H_0$ by $H_0-\kappa I$ we can incorporate a chemical potential in our results. For simplicity, we shall make no further reference to $\kappa$ during the proofs and assume it was already incorporated to $H_0$.
\end{Rem}
\section{Regularity of the effective potential}
\label{regularitysection}
Our goal in this section is to conclude that, for a fixed finite subset $\Lambda \subset \mathbb{Z}^d$ with $|\Lambda|=n$, the effective potential $V_{\mathrm{eff}}$ is a smooth function of $\{\omega(j)\}_{j\in \Lambda}$. This will be of relevance for several resampling arguments later in the note. For that purpose, define a map $\xi:\ell^{\infty}\left(\mathbb{Z}^d \right)\times \mathbb{R}^{n}\rightarrow \ell^{\infty}\left(\mathbb{Z}^d \right)$ by 
\begin{equation}\xi(V,\omega)(j)=V(j)-\mel{j}{F(H_0+\lambda\omega+gV)}{j}
\end{equation}
Then, $V_{\mathrm{eff}}$ is the unique solution of $\xi(V,\omega)=0$.
 Thus, its regularity can inferred via the implicit function theorem once we check that the derivative $D\xi(\cdot,\omega)$ is non-singular. Note that
\begin{equation}
\frac{\partial \xi(V,\omega)(j)}{\partial V(l)}=
\delta_{jl}-\frac{\partial \mel{j}{F(H_0+\lambda\omega+gV)}{j}}{\partial V(l)}.
\end{equation}
Using lemma \ref{Contraction}, we have that
\begin{equation}
\Big|\frac{\partial \mel{j}{F(H_0+\lambda\omega+gV)}{j}}{\partial V(l)}\Big|\leq 
|g|\frac{72\sqrt{2}e^{-2\nu|j-l|}}{\eta}\|F\|_{\infty}.
\end{equation}

In particular, whenever $|g|\frac{72\sqrt{2}\|F\|_{\infty}}{\eta(1-e^{-2\nu})^{d}}<1$  we have that the operator $D\xi(\omega,.):\ell^{\infty}\left(\Lambda\right)\rightarrow \ell^{\infty}\left(\Lambda \right)$ is invertible since it has the form $I+gM$ where $gM$ has operator norm less than one. Note the smallness condition on $g$ is independent of $\Lambda\subset \mathbb{Z}^d$. It is a consequence of the implicit function theorem that $V$ is a smooth function of $(\omega(1),...,\omega(n))$.

\section{Decay estimates for the effective potential}
\label{decaysection}

We start this section with the following lemma, which will be useful to formulate the decay of correlations between $U(n)$ and $U(m)$ as $|m-n|\to \infty$.

\begin{lem}\label{decay1} Whenever $\frac{72\sqrt{2}|g|\|F\|_{\infty}}{\eta\left(1-e^{-\nu}\right)^d}<1$, there exist constants $C_1(d,\lambda,g,\eta,\|F\|_{\infty},\nu)$ and $C_2(d,\lambda,g,\eta,\|F\|_{\infty},\nu)$ such that
\begin{equation}\label{lemmapotential1}\max\Big\{\sum_{m}e^{\nu|n-m|}\Big|\frac{\partial V_{\mathrm{eff}}(n)}{\partial \omega(m)}\Big|,\sum_{n}e^{\nu|n-m|}\Big|\frac{\partial V_{\mathrm{eff}}(n)}{\partial \omega(m)}\Big|\,\Big\}\leq C_1
\end{equation}

\begin{equation}\label{lemmapotential2}
\sum_{l,m,n}e^{\nu\left(|l-n|+|n-m|+|l-m|\right)}\Big|\frac{\partial^2V_{\mathrm{eff}}(n)}{\partial \omega(m)\partial \omega(l)}\Big|\leq C_2.
\end{equation}
Moreover $C_1$ and $C_2$ can be bounded from above by a constant of the form $\frac{\lambda D}{1-g\theta}$ with $D$ and $\theta$ independent of $g$ and these constants are explicit in the proof.

\end{lem}
\begin{proof}

For convenience we denote $V_{\mathrm{eff}}=V$. As in section \ref{existencesection} we write $F(H)=\frac{1}{2\pi i}\int^{\infty}_{-\infty}\left(\frac{1}{H+t-i\eta}-\frac{1}{H+t+i\eta}\right) f(t)\,dt$ where $f$ is bounded by $3\|F\|_{\infty}$. Thus
\begin{equation}\label{intrep}
V(n,\omega)=\frac{1}{2\pi i}\int^{\infty}_{-\infty} K(n,t,\omega)f(t)\,dt
\end{equation}
where $K(n,t,\omega)=G(n,n;t-i\eta)-G(n,n;t+i\eta)$. Denote by $P_{m}$ the projection mapping $\ell^2\left(\mathbb{Z}^d\right)$ onto $\ell^2\left(\mathrm{Span}\{\delta_m\}\right)$. Using difference quotients, it is easy to check
 \begin{equation}\label{ResDeriv}
\frac{\partial}{\partial \omega(m)}\frac{1}{H-z}+ g \frac{1}{H-z}\frac{\partial V}{\partial \omega(m)}\frac{1}{H-z}=- \lambda\frac{1}{H-z}P_{m}\frac{1}{H-z}.
\end{equation}
 Taking matrix elements we obtain

$$\frac{\partial K(n,t,\omega)}{\partial \omega(m)}=-g\sum_{l}\tilde{G}(l,n)\frac{\partial V(l)}{\partial \omega(m)}+
\lambda r(m,n). $$
$$\tilde{G}(l,n):=G(l,n;t+i\eta)G(n,l;t+i\eta)-G(l,n;t-i\eta)G(n,l;t-i\eta).$$
 $$r(m,n):=G(n,m;t+i\eta)G(m,n;t+i\eta)-G(n,m;t-i\eta)G(m,n;t-i\eta).$$
 Note 
 \begin{equation}\label{rewriteG}\tilde{G}(l,n)=\left(G(l,n;t+i\eta)-G(l,n;t-i\eta)\right)G(n,l;t+i\eta)+\left(G(n,l;t+i\eta)-G(n,l;t-i\eta)\right)G(l,n;t-i\eta).
 \end{equation}
 We now make use of \cite[Lemma 3]{A-Graf}:
 \begin{equation}\label{cancellemma}|G(l,n;t+i\eta)-G(l,n;t-i\eta)|\leq 12\eta e^{-\nu|l-n|}\mel{n}{\frac{1}{(H-t)^2+\eta^2/2}}{n}^{1/2}\mel{l}{\frac{1}{(H-t)^2+\eta^2/2}}{l}^{1/2}.
  \end{equation}
  This, together with the Combes-Thomas bound $|G(l,n,t\pm i\eta)|\leq \frac{2}{\eta}e^{-\nu|l-n|}$ and (\ref{rewriteG}) implies
  $$|\tilde{G}(l,n)|\leq 48e^{-2\nu|l-n|}\mel{n}{\frac{1}{(H-t)^2+\eta^2/2}}{n}^{1/2}\mel{l}{\frac{1}{(H-t)^2+\eta^2/2}}{l}^{1/2}.$$
  $$|r(m,n)|\leq 48e^{-2\nu|m-n|}\mel{m}{\frac{1}{(H-t)^2+\eta^2/2}}{m}^{1/2}\mel{n}{\frac{1}{(H-t)^2+\eta^2/2}}{n}^{1/2}.$$
Thus
$$\tilde{K}(l,n):=\int^{\infty}_{-\infty}|\tilde{G}(l,n)|\,dt\leq\frac{48\sqrt2 \pi}{\eta}e^{-2\nu|l-n|}$$
$$\tilde{r}(m,n):=\int^{\infty}_{-\infty} |r(m,n)|\,dt\leq\frac{48\sqrt2 \pi}{\eta}e^{-2\nu|m-n|}.$$
To summarize, we have shown the following inequality
$$\Big|\frac{\partial V(n)}{\partial \omega(m)}\Big |\leq \frac{3\|F\|_{\infty}}{2\pi}\left(|g|\sum_{l}\tilde{K}(l,n)\Big|\frac{\partial V(l)}{\partial \omega(m)}\Big|+
\lambda\tilde{r}(m,n)\right). $$
Whenever $\frac{72\sqrt{2}|g|}{\eta\left(1-e^{-2\nu}\right)^d}<1$ we have that
\begin{equation}\label{smallkernel}
|g|\|\tilde{K}\|_{\infty,\infty}<1
\end{equation}
 where
 
 \begin{equation} \label{Ksum}\|\tilde{K}\|_{\infty,\infty}=\sup_{l}\sum_{m}\tilde{K}(l,m).
 \end{equation}
Considering the weight $W(n):=e^{\nu|m-n|}$ we let
\begin{equation}\theta:=\frac{3\|F\|_{\infty}}{2\pi}\sup_{n}\sum_{l}\frac{W(n)}{W(l)}\tilde{K}(n,l).
\end{equation}
By the triangle inequality,
\begin{align*}\\
\theta&\leq \frac{3\|F\|_{\infty}}{2\pi}\sup_{n}\sum_{l}e^{\nu|n-l|}\tilde{K}(n,l)\\
&\leq \frac{72\sqrt2 \|F\|_{\infty}}{\eta\left(1-e^{-\nu}\right)^d}.\\
\end{align*}
hence, whenever $\frac{72\sqrt{2}|g|}{\eta\left(1-e^{-\nu}\right)^d}<1$, we have that
\begin{equation}\label{smallkernel2} |g|\theta<1.
\end{equation}
Moreover, with the choice
\begin{equation}D_1:=\sum_{n}W(n)\tilde{r}(m,n)
\end{equation}
we have
\begin{equation}\label{smallkernel3}
D_1\leq \frac{72\sqrt2 \|F\|_{\infty}}{\eta\left(1-e^{-\nu}\right)^d}.
\end{equation}

After conditions (\ref{smallkernel}), (\ref{smallkernel2}) and (\ref{smallkernel3}) have been verified, the general result \cite[Theorem 9.2]{A-W-B} applies, yielding
\begin{equation}\sum_{m}e^{\nu|n-m|}\Big|\frac{\partial V(n)}{\partial \omega(m)}\Big|<\frac{\lambda D_1}{1-g\theta}:=C_1(d,\|F\|_{\infty},\lambda,g,\eta,\nu).
\end{equation}

Differentiating (\ref{ResDeriv}) with respect to $\omega(l)$,
\begin{align*}
\frac{\partial^2}{\partial \omega(m)\partial \omega(l)}\frac{1}{H-z}&+ g \left(\frac{\partial}{\partial \omega(l)}\frac{1}{H-z}\right)\frac{\partial V}{\partial \omega(m)}\frac{1}{H-z}
+g \frac{1}{H-z}\frac{\partial V}{\partial \omega(m)}\left(\frac{\partial}{\partial \omega(l)}\frac{1}{H-z}\right)\\
+g\frac{1}{H-z}\frac{\partial^2 V}{\partial \omega(m)\partial \omega(l)}\frac{1}{H-z}
&=-\lambda\left(\frac{\partial}{\partial \omega(l)} \frac{1}{H-z}\right)P_{m}\frac{1}{H-z}-\lambda\frac{1}{H-z}P_{m}\left(\frac{\partial}{\partial \omega(l)}\frac{1}{H-z}\right)\\
\end{align*}

Repeating the previous argument and using the established decay of $\frac{\partial V(n)}{\partial \omega(m)}$ we reach (\ref{lemmapotential2}), finishing the proof.
\end{proof}

Given a finite set $\Lambda \subset \mathbb{Z}^d$, let us define $\mathcal{T}:\mathbb{R}^{|\Lambda|} \rightarrow \mathbb{R}^{|\Lambda|}$ by
\begin{equation}
\left(\mathcal{T}\omega\right)(n)=\omega(n)+\frac{g}{\lambda}V_{\mathrm{eff}}(n).
\end{equation}
Let $U(n):=\left(\mathcal{T}\omega\right)(n)$ be the new coordinates in the probability space. The bound (\ref{lemmapotential1}) implies that, for $|g|$ sufficiently small, $\mathcal{T}$ is a differentiable perturbation of the identity by an operator with norm less than one hence $\mathcal{T}^{-1}$ is well defined.\par Fix $n_0 \in \Lambda$ and denote by $U_{\alpha}=U+\left(\alpha-U(n_0)\right)\delta_{n_0}$ the new potential  obtained from $U$ by setting its value at $n_0$ equal to $\alpha$. Let $\omega_{\alpha}(n)=\left(\mathcal{T}^{-1}U_{\alpha}\right)(n)$. The variables $\omega_{\alpha}(n)$ correspond to the change in $\omega(n)$ when a resampling argument is applied to the new probability space at the point $n_0$. Intuitively, the exponential decay guarantees that this change is not too large if $n$ and $n_0$ are far away. This is the content of the lemma below.
\begin{lem}\label{resamp} For all $\alpha\in \mathbb{R}$ and $|g|<\lambda C^{-1}_1$, we have
$$
\sum_{n\neq n_0} e^{\nu |n-n_0|}\big|\omega_\alpha(n)-\omega(n)\big|\leq
\frac{C_1|g|}{\lambda}\frac{\left(|\alpha-U(n_0)|+2\frac{|g|\|F\|_{\infty}}{\lambda}\right)}{\left(1-\frac{|g|}{\lambda}C_1\right)}.
$$
where $C_1$ is the upper bound on equation (\ref{lemmapotential1}).

\end{lem}
\begin{proof}
Using the given definitions and the mean value inequality we obtain, for $n\neq n_0$,
\begin{align*}
|\omega(n)-\omega_{\alpha}(n)|&\leq \frac{|g|}{\lambda}|V(n,\omega)-V(n,\omega_{\alpha})|\\
&\leq\frac{|g|}{\lambda}\sum_{l\in \mathbb{Z}^d} \Big|\frac{\partial V_{\mathrm{eff}}(n,\hat{\omega}_{\alpha})}{\partial \omega(l)}\Big|\Big|\omega_{\alpha}(l)-\omega(l)\Big|\\
&\leq \frac{|g|}{\lambda}\Big|\frac{\partial V_{\mathrm{eff}}(n,\hat{\omega}_{\alpha})}{\partial \omega(n_0)}\Big|\left(|\alpha-U(n_0)|+2\frac{|g|\|F\|_{\infty}}{\lambda}\right)+\frac{|g|}{\lambda}\sum_{l\neq n_0} \Big|\frac{\partial V_{\mathrm{eff}}(n,\hat{\omega}_{\alpha})}{\partial \omega(l)}\Big|\Big|\omega_{\alpha}(l)-\omega(l)\Big|.
\end{align*}
Where $\hat{\omega}_{\alpha}$ denotes some configuration   with $\hat{\omega}_{\alpha}(l)$ in the interval connecting $\omega(l)$ to  $\omega_{\alpha}(l)$.
Let $W(n)=e^{\nu|n-n_0|}$. According to (\ref{lemmapotential1}),
\begin{align*}\sup_n\sum_l \frac{W(n)}{W(l)}\Big|\frac{\partial V_{\mathrm{eff}}(n,\hat{\omega}_{\alpha})}{\partial \omega(l)}\Big|&\leq\sup_n\sum_l e^{\nu|n-l|}\Big|\frac{\partial V_{\mathrm{eff}}(n,\hat{\omega}_{\alpha})}{\partial \omega(l)}\Big|\\
&\leq C_1.\\
\end{align*}
 Once again, the conditions of \cite[Theorem 9.2]{A-W-B} are satisfied for 
$|g|<\lambda C^{-1}_1$
therefore
$$\sum_{n\neq n_0} e^{\nu|n-n_0|}\big|\omega_{\alpha}(n)-\omega(n)\big|\leq \frac{C_1|g|}{\lambda}\frac{\left(|\alpha-U(n_0)|+2\frac{|g|\|F\|_{\infty}}{\lambda}\right)}{\left(1-\frac{|g|}{\lambda}C_1\right)}.\qedhere$$

\end{proof}
Since another application of the mean value theorem gives, after a possible correction on $\hat{\omega}_{\alpha}$ that \begin{align*} \Big|\frac{\partial V_{\mathrm{eff}}(n,\omega)}{\partial \omega(m)}-\frac{\partial V_{\mathrm{eff}}(n,\omega_{\alpha})}{\partial \omega(m)}\Big|&\leq
\sum_{l\in \mathbb{Z}^d}\Big| \frac{\partial^2 V_{\mathrm{eff}}(n,\hat{\omega}_{\alpha})}{\partial \omega(l)\partial \omega(m)}\Big|\Big|\omega(l)-\omega_{\alpha}(l)\Big|\\
\end{align*}
we obtain, for any $\nu'\in(0,\nu)$,
\begin{align*}\Big|\frac{\partial V_{\mathrm{eff}}(n,\omega)}{\partial \omega(m)}-\frac{\partial V_{\mathrm{eff}}(n,\omega_{\alpha})}{\partial \omega(m)}\Big|\leq&\frac{C_2|g|}{\lambda}\left(C_1\frac{\left(|\alpha-U(n_0)|+2\frac{|g|\|F\|_{\infty}}{\lambda}\right)}{\left(1-\frac{|g|}{\lambda}C_1\right)(1-e^{\nu'-\nu})^d}+2\|F\|_{\infty} \right)e^{-\nu'\left(|m-n|+|n-n_0|+|m-n_0|\right)}.
\end{align*}
where $C_2$ is the constant in (\ref{lemmapotential2}).

In particular, letting $\nu'=\nu/2$, if $A=\frac{g}{\lambda}\left(\frac{\partial V_{\mathrm{eff}}(n_i,\omega_{\alpha})}{\partial \omega(n_j)}\right)_{|\Lambda|\times |\Lambda|}$ and $B=\frac{g}{\lambda}\left(\frac{\partial V_{\mathrm{eff}}(n_i,\omega)}{\partial \omega(n_j)}\right)_{|\Lambda|\times |\Lambda|}$ we have
\begin{equation}\label{explicitboundsum}\sum_{(m,n)\in \Lambda \times \Lambda}|(A-B)_{m,n}|\leq \frac{C_2|g|^2}{\lambda^2}\left(C_1\frac{\left(|\alpha-U(n_0)|+2\frac{|g|\|F\|_{\infty}}{\lambda}\right)}{\left(1-\frac{|g|}{\lambda}C_1\right)(1-e^{-\nu/2})^{3d}}+\frac{2\|F\|_{\infty}}{(1-e^{-\nu/2})^{2d}} \right).
\end{equation} 
We summarize the above observation as a lemma.
\begin{lem}\label{bddtrace}
Let $A=\frac{g}{\lambda}\left(\frac{\partial V_{\mathrm{eff}}(n_i,\omega_{\alpha})}{\partial \omega(n_j)}\right)_{|\Lambda|\times |\Lambda|}$ and $B=\frac{g}{\lambda}\left(\frac{\partial V_{\mathrm{eff}}(n_i,\omega)}{\partial \omega(n_j)}\right)_{|\Lambda|\times |\Lambda|}$. Whenever $ \frac{|g|}{\lambda}C_1<1$ we have
\begin{equation}
  \sum_{(m,n)\in \Lambda \times \Lambda}|(A-B)_{m,n}|\leq  |g|^2\left(C_3|\alpha-U(n_0)|+C_4\right).
\end{equation}
Moreover, the constant $C_3$ can be chosen independent of $\lambda$ and $C_4$ is proportional to $\frac{1}{\lambda}$.
\end{lem}

Finally, we  analyze how the effective potential varies with respect to disorder and interaction. This will be relevant to the Integrated Density of States regularity. More precisely
\begin{lem}\label{comparepotentials} For a fixed $\omega \in \Omega$
\begin{equation} 
|V_{\lambda,g}(n)-V_{\lambda',g'}(n)|\leq \frac{C_5(d,\|F\|_{\infty},g,\eta,\nu,\omega)}{1-gC_6(d,\|F\|_{\infty},g,\eta,\nu)}|\lambda-\lambda'| +C_7(d,\|F\|_{\infty},g,\eta,\nu)|g-g'|.
\end{equation} 
\end{lem}
Note when $\lambda\neq \lambda'$ the bound depends on $\omega$ through the constant $C_5$.
\begin{proof}
Let $R_{\lambda,g}(z)=\frac{1}{H_0+\lambda \omega+gV_{\lambda,g}-z}$ and $R_{\lambda',g'}(z)=\frac{1}{H_0+\lambda' \omega+g'V_{\lambda',g'}-z}$ for $z=t+i\eta$. Similarly as in the above proofs, it is immediate to check that
\begin{equation}
R_{\lambda,g}(z)-R_{\lambda',g'}(z)=(\lambda-\lambda')R_{\lambda,g}(z)V_{\omega}R_{\lambda',g'}(z)+(g-g')R_{\lambda,g}(z)V_{\lambda',g'}R_{\lambda',g'}(z)
-gR_{\lambda,g}(z)\left(V_{\lambda,g}-V_{\lambda',g'}\right)R_{\lambda',g'}(z).
\end{equation}
 Replacing $z$ by $\bar{z}$ and subtracting the resulting equations:
 \begin{align*}\left(R_{\lambda,g}(z)-R_{\lambda,g}(\bar{z})\right)-\left(R_{\lambda',g'}(z) -R_{\lambda',g'}(\bar{z})\right)=&
 \left(R_{\lambda,g}(z)-R_{\lambda,g}(\bar{z})\right)\left((\lambda-\lambda')V_{\omega}+(g-g')V_{\lambda',g'}\right)R_{\lambda',g'}(z) \\
 &+R_{\lambda,g}(z)\left((\lambda-\lambda')V_{\omega}+(g-g')V_{\lambda',g'}\right)\left(R_{\lambda',g'}(z)-R_{\lambda',g'}(\bar{z})\right)\\
 &-gR_{\lambda,g}(z)\left(V_{\lambda,g}-V_{\lambda',g'}\right)\left(R_{\lambda',g'}(z)-R_{\lambda',g'}(\bar{z})\right).
 \end{align*}
 Taking matrix elements, multiplying by $f(t)$, integrating with respect to $t$ and taking absolute values we can read from the representation (\ref{intrep}) that, denoting
  \begin{equation}\label{differencekernel}
  K_{\lambda,g}(n,l)=|G_{\lambda,g}(n,l;z)-G_{\lambda,g}(n,l;\bar{z})|,
 \end{equation}
 \begin{align*}|V_{\lambda,g}(n)-V_{\lambda',g'}(n)|\leq&\frac{3\|F\|_{\infty}}{2\pi}|\lambda-\lambda'|\sum_{l\in \mathbb{Z}^d}|\omega(l)|\int^{\infty}_{-\infty}\left({{K}}_{\lambda,g}(n,l){{K}}_{\lambda',g'}(l,n)+|G_{\lambda,g}(n,l)|{\tilde{K}}_{\lambda',g'}(l,n)\right)\,dt\\
 &+\frac{3\|F\|^2_{\infty}}{2\pi}|g-g'|\sum_{l\in \mathbb{Z}^d}\int^{\infty}_{-\infty}\left(|G_{\lambda,g}(n,l)|{{K}}_{\lambda',g'}(l,n)+G_{\lambda',g'}(n,l)|{{K}}_{\lambda,g}(l,n) \right)\,dt.\\
 &+g\sum_{l\in \mathbb{Z}^d}\int^{\infty}_{-\infty}|G_{\lambda,g}(n,l)||V_{\lambda,g}(l)-V_{\lambda',g'}(l)|{{K}}_{\lambda',g'}(l,n)\,dt.
 \end{align*}
  Using equation \ref{cancellemma} together with \cite[Theorem 9.2]{A-W-B} we conclude the proof.
\end{proof}

\subsection{Improvements}\label{improvesec}
We will now improve upon the previous bounds. Specifically, we need robust bounds which also reflect the decay of the derivatives of $V_{\mathrm{eff}}(n)$ when the local potential $\omega(n)$ is large.
 The improvements on this section will be important for a general fluctuation analysis on section \ref{proofoflemma} and for localization in the one dimensional setting. 
 Before stating the main result of the section we start with the following deterministic estimate, which incorporates ideas from \cite[Lemma 3]{A-Graf}.
 \begin{lem}\label{adaptAG}
 \begin{equation}|G(m,l;t+i\eta)|\leq \sqrt{2}\mel{l}{\frac{1}{(H-t)^2+\eta^2}}{l}^{1/2}e^{-\nu|m-l|}
 \end{equation}
 \end{lem}
 \begin{proof} To keep the notation simple, we set $t=0$ without loss of generality. Let $H_{f}=e^{f}He^{-f}$ where $f(n)=\nu\min\{|n-l|,R\}$ for a fixed $l\in \mathbb{Z}^d$ and $R>0$. By choosing $R$ sufficiently large, we may assume that $|m-l|<R$. We then have
 \begin{equation*}\label{trick}e^{\nu|m-l|}G(m,l;i\eta)=\mel{m}{(H_f-i\eta)^{-1}}{l}.
 \end{equation*}
 We claim that
 \begin{equation}\label{claimAG}
 ||(H_f-i\eta)^{-1}(H^2+\eta^2)^{1/2}||\leq \sqrt{2}.
 \end{equation}
 Indeed,
 \begin{align*}
 ||(H_f-i\eta)^{-1}(H^2+\frac{\eta^2}{2})^{1/2}||^2&=||(H^2+\frac{\eta^2}{2})^{1/2}(H^{\ast}_f+i\eta)^{-1}(H_f-i\eta)^{-1}(H^2+\frac{\eta^2}{2})^{1/2}||\\
 &=||(H^2+\frac{\eta^2}{2})^{1/2}\frac{1}{(H_f-i\eta)(H^{\ast}_f-i\eta)}(H^2+\frac{\eta^2}{2})^{1/2}|| 
 \end{align*}
 where by \cite[Eq D.9]{A-Graf} (with $f$ replaced by $-f$) we have
 \begin{equation*}(H_f-i\eta)(H^{\ast}_f-i\eta)\geq \frac{1}{2}\left( H^2+\frac{\eta^2}{2}\right)
 \end{equation*} 
 showing the claim in (\ref{claimAG}). Equation (\ref{adaptAG}) will now follow from
\begin{align*}|\mel{m}{(H_f-i\eta)^{-1}}{l}|&\leq \|(H_f-i\eta)^{-1}(H^2+\frac{\eta^2}{2})^{1/2}\|\,|(H^2+\frac{\eta^2}{2})^{-1/2}\delta_l|\\
&\leq\sqrt{2}\mel{l}{(H^2+\frac{\eta^2}{2})^{-1}}{l}^{1/2}.\qedhere
\end{align*} 
 \end{proof}

\begin{lem}\label{improvedct} There exists $C_7(\lambda,\eta,d,g,\|F\|_{\infty},\nu)>0$ such that, for $m\neq n$, 
\begin{equation}\label{improvedct1}
\max\{\,|\omega(n)|,|\omega(m)|\,\}\Big|\frac{\partial V(n)}{\partial \omega(m)}\Big|\leq C_7e^{-2\nu|m-n|}
\end{equation}

and, for $n\neq n_0$,
\begin{equation}\label{improvedct2}\big|\omega(n)(\omega_{\alpha}(n)-\omega(n))\big|\leq
\frac{C_7|g|}{\lambda-|g|C_1}\left(|\alpha-U(n_0)|+2\frac{|g|\|F\|_{\infty}}{\lambda}+\frac{1}{\left(1-e^{-\nu}\right)^d}\right)e^{-\nu|n-n_0|}.
\end{equation}
Moreover, whenever $\frac{|g|}{\lambda}C_1<1$,  $C_7$ can be chosen to be uniformly bounded as a function of the parameters $\lambda$ and $g$.

\end{lem}
\begin{proof}
Recall that $U(n)=\omega(n)+\frac{g}{\lambda}V_{\mathrm{eff}}(n)$ denotes the ``full" potential at site $n$. We split the proof in two cases.
\setlist[enumerate,1]{label={(\roman*)}}
\begin{enumerate}
\item{ Case one: $U(n)\geq 0$.\par
 Let us start by noting that lemma \ref{adaptAG} implies that for $n,l \in \mathbb{Z}^d$
 \begin{equation*}
 \int^{\infty}_{-\infty}|G(n,l;t+i\eta)G(l,n;t+i\eta)|\,dt\leq\frac{2\sqrt{2}\pi}{\eta}e^{-2\nu|n-l|}.
 \end{equation*}

From the previous section we already know that
\begin{equation}\label{repderv}
\frac{\partial V(n)}{\partial \omega(m)}=\int^{\infty}_{-\infty}\left( -g\sum_{l}r(n,l)\frac{\partial V(l)}{\partial \omega(m)}+
\lambda r(m,n)\right)f(t)\,dt
\end{equation}
where $f(t)=F_{+}(t+i\eta)+F_{-}(t-i\eta)-D\ast F(t)$ and
$$r(m,n)=G(n,m;t+i\eta)G(m,n;t+i\eta)-G(n,m;t-i\eta)G(m,n;t-i\eta).$$
 Observe that, for $z=t+i\eta$ and $n\neq m$,
\begin{align*}\lambda|U(n)G(n,m;t+i\eta)G(m,n;t+i\eta)|
&=\frac{\lambda|U(n)|}{|\lambda U(n)-z|}\sum_{l} |H_{0}(n,l)G(l,m;t+i\eta)G(m,n,t+i\eta)|\\
&\leq \left(1+\frac{|z|}{|\lambda U(n)-z|}\right)\sum_{l}
|H_{0}(n,l)G(l,m;t+i\eta)G(m,n,t+i\eta)|
\end{align*}
where we made use of the identity

\begin{equation}\label{depleted}
(\lambda U(m)-z)G(n,m;z)=\delta_{mn}-\sum_{l}H_{0}(n,l) G(l,m;z).
\end{equation}
Note that if $U(n)\geq 0$ and $t=\mathrm{Re}z<0$, then
\begin{equation}\frac{|z|}{|\lambda U(n)-z|}\leq 1.
\end{equation}
Using the fact that $tf(t)$ goes to zero as $t\to \infty$ we conclude that
\begin{equation}
\int^{\infty}_{-\infty}|U(n)G(n,m;t+i\eta)G(m,n;t+i\eta)||f(t)|\,dt\leq \frac{C(\nu,\eta,\|F\|_{\infty})}{\lambda}e^{-2\nu|m-n|}.
\end{equation}
 Since a similar equation holds with $m$ replaced by $l$, we can proceed as in the previous section
 and, using the exponential decay of $\frac{\partial V(n)}{\partial \omega(m)}$, conclude the proof.}
 \item{ Case two: $U(n)<0$.
 
 In this case, the argument given above must be modified to take into account that the inequality \begin{equation}\frac{|z|}{|\lambda U(n)-z|}\leq 1.
\end{equation}
is satisfied when $t=\mathrm{Re}>0$. In this case, the use of (\ref{repderv}) would result in a problem as $tf(t)$ is unbounded as $t\to -\infty$. This can be addressed by observing that the Fermi-Dirac function $F(z)$ has the following symmetry
\begin{equation}\frac{1}{1+e^{\beta(z-\mu)}}=1-\frac{1}{1+e^{\beta(-z+\mu)}}.
\end{equation}
Hence we can make use of the representation (\ref{repderv}) corresponding to
\begin{equation}
-\frac{1}{1+e^{\beta(-z+\mu)}}:=F^{\ast}_{\mu}(z)
\end{equation}
since, for $m\neq n$, the constant term does not affect the calculation of $\frac{\partial V(n)}{\partial \omega(m)}$. Denoting by \begin{equation*}f^{\ast}(t)=F^{\ast}_{+}(t+i\eta)+F^{\ast}_{-}(t-i\eta)-D\ast F^{\ast}(t)
\end{equation*}
 we reach
\begin{equation}\label{repderv2}
\frac{\partial V(n)}{\partial \omega(m)}=\int^{\infty}_{-\infty}\left( -g\sum_{l}r(n,l)\frac{\partial V(l)}{\partial \omega(m)}+
r(m,n)\right)f^{\ast}(t)\,dt
\end{equation}
where now $tf^{\ast}(t)\to 0$ as $t\to -\infty$. Proceeding as in the first case the proof is finished, showing (\ref{improvedct1}). Following the proof of lemma \ref{resamp} and using (\ref{improvedct1}) we conclude (\ref{improvedct2})\qedhere}

 \end{enumerate}
 \end{proof}
 \begin{lem}\label{localbeh} Let $\Lambda_1$ and $\Lambda_2$ be subsets of $\mathbb{Z}^d$ with $dist(\Lambda_1,\Lambda_2)\geq r$. Let $V^{c}_{2}$ be the effective potential defined by
 \begin{equation*}
 V^{c}_{2}(n)=\mel{n}{F(H^{\Lambda^{c}_{2}})}{n}\,\,\,n\in \mathbb{Z}^d.
 \end{equation*}
 where $H^{\Lambda^{c}_{2}}$ denotes the restriction of $H$ to the complement of $\Lambda_2$.
 
 Then, for any $n\in \Lambda_1$
 \begin{equation}
 \Big|\frac{\partial V(n)}{\partial \omega(m)}-\frac{\partial V^{c}_2(n)}{\partial \omega(m)}\Big|\leq C(\eta,d,\lambda,g,\|F\|_{\infty},\nu)e^{-\nu(|m-n|+r)}
 \end{equation}
 
 \end{lem}
 \begin{proof}
 The proof follows the same steps as in the previous results. The only modification which is required comes when comparing the quantities $r(m,n)$ and $r^{\Lambda_1}(m,n)$ given by
 $$r(m,n)=G(n,m;t+i\eta)G(m,n;t+i\eta)-G(n,m;t-i\eta)G(m,n;t-i\eta)$$
 $$r^{\Lambda^{c}_{2}}(m,n)=G^{\Lambda^{c}_{2}}(n,m;t+i\eta)G^{\Lambda^{c}_{2}}(m,n;t+i\eta)-G^{\Lambda^{c}_{2}}(n,m;t-i\eta)G^{\Lambda^{c}_{2}}(m,n;t-i\eta).$$
We observe that
\begin{align*}
&G(n,m;z)G(m,n;z)-G^{\Lambda^{c}_{2}}(n,m;z)G^{\Lambda^{c}_{2}}(m,n;z)=\\
&G(n,m;z)\left(G(m,n;z)-G^{\Lambda^{c}_{2}}(m,n;z)\right)+\left(G(n,m;z)-G^{\Lambda^{c}_{2}}(n,m;z)\right)G^{\Lambda^{c}_{2}}(m,n;z).
\end{align*}
Moreover,
\begin{equation}
G(m,n;z)-G^{\Lambda^{c}_{2}}(m,n;z)=-\lambda\sum_{l\in {\Lambda_{2}}}G(m,l;z)U(l)G^{\Lambda^{c}_{2}}(l,n;z).
\end{equation}
The proof is now finished using arguments identical to the proof of lemma \ref{decay1} and the improvement on lemma \ref{improvedct}.

 \end{proof}

\section{Proof of lemma \ref{bdddensity} }  
\label{proofoflemma}
     In this section we show the existence of the effective density $\rho_{\mathrm{eff}}$. Fix $\Lambda \subset \mathbb{Z}^d$ finite. Recall that we defined
    \begin{equation}\label{U}
 U(n,\omega):=\omega(n)+\frac{g}{\lambda}V_{\mathrm{eff}}(n,\omega).
 \end{equation}
 Until the end of this section we suppress the $\omega$ dependence on $U(n)$ and $V_{\mathrm{eff}}$.
Note that, for $m,n\in \Lambda$,
 \begin{equation}\label{DerU}
 \frac{\partial U(m)}{\partial \omega(n)}=\delta_{mn}+\frac{g}{\lambda} \frac{\partial V_{\mathrm{eff}}(m)}{\partial \omega(n)}.
 \end{equation}
 We have denoted the above change of variables by $\mathcal{T}:\mathbb{R}^{|\Lambda|} \rightarrow \mathbb{R}^{|\Lambda|}$, which reads
\begin{equation}\label{DefT} \mathcal{T}(\omega(n_1),...,\omega(n_{|\Lambda|}))=(U(n_1),...,U(n_{|\Lambda|}))
\end{equation}

We can now compute the joint distribution of the $\{U(n)\}$. Using the fact that the random variables $\{\omega(n)\}_{n\in \mathbb{Z}^d}$ have a common density $\rho$ we conclude that for all Borel sets $I_1,...,I_N$ in $\mathbb{R}$:
\begin{align*} \mathbb{P}\left( U(n_1)\in I_1,\,...\,,U(n_{|\Lambda|}),\in I_{|\Lambda|}\right)&=\int_{{\mathcal{T}}^{-1}\left(I_1\times...\times I_{|\Lambda|}\right)}\prod^{|\Lambda|}_{k=1}\rho(\omega(n_k))\,d\omega(n_k)\\
&=\int_{I_1\times...\times I_{|\Lambda|}}|\det J_{{\mathcal{T}}^{-1}}|\prod^{|\Lambda|}_{k=1}\rho\left({\mathcal{T}}^{-1}U(n_k)\right)\,dU(n_1)...dU(n_{|\Lambda|})\\
&=\int_{I_1\times...\times I_{|\Lambda|}}\big|\det\Big{(}I+\frac{g}{\lambda}\frac{\partial V_{\mathrm{eff}}(n_i,{\mathcal{T}}^{-1}U)}{\partial U(n_j)}\Big{)}\big|\prod^{|\Lambda|}_{k=1}\rho\left(U(n_k)-\frac{g}{\lambda}V_{\mathrm{eff}}(n_k,{\mathcal{T}}^{-1}U)\right)\,dU(n_k).
\end{align*}

Therefore the joint distribution of $\{U(n_k)\}^{|\Lambda|}_{k=1}$ is given by the measure  \begin{equation}\label{densU}\Big|\det\Big{(}I+\frac{g}{\lambda}\frac{\partial V_{\mathrm{eff}}(n_i,{\mathcal{T}}^{-1}U)}{\partial U(n_j)}\Big{)}\Big|\prod^{|\Lambda|}_{k=1}\rho\left(U(n_k)-\frac{g}{\lambda}V_{\mathrm{eff}}(n_k,\mathcal{T}^{-1}U) \right)\,dU(n_1)...dU(n_{|\Lambda|}). \end{equation}
It follows that for each $n_0\in \Lambda$ the conditional expectation of $U(n_0)$ at specified values of $\{U(n)\}_{n\neq n_0}$ has a density given by
\begin{equation}\label{cdensU}{\rho}^{\Lambda}_{n_0}=\frac{\big|\det\Big{(}I+\frac{g}{\lambda}\frac{\partial V_{\mathrm{eff}}(n_i,{\mathcal{T}}^{-1}U)}{\partial U(n_j)}\Big{)}\big|\prod^{|\Lambda|}_{k=1}\rho\left(U(n_k)-\frac{g}{\lambda}V_{\mathrm{eff}}(n_k,{\mathcal{T}}^{-1}U)\right)}{\int^{\infty}_{-\infty}\big|\det\Big{(}I+\frac{g}{\lambda}\frac{\partial V_{\mathrm{eff}}(n_i,{\mathcal{T}}^{-1}U^{\alpha})}{\partial U(n_j)}\Big{)}\big|\prod^{|\Lambda|}_{k=1}\rho\left(U^{\alpha}(n_k)-\frac{g}{\lambda}V_{\mathrm{eff}}(n_k,{\mathcal{T}}^{-1}U^{\alpha}) \right)\,d\alpha} \end{equation}
Where $U^{\alpha}(n):=U(n)+\left(\alpha-U(n_0)\right)\delta_{n=n_0}$.
This strategy naturally leads to the analysis of ratios of determinants. A sufficient condition for an upper bound to the right-hand side of (\ref{cdensU}) is to obtain positive constants $C=C_{\mathrm{fluct}}(U(n_0))$ and $D=D(\alpha)$ which are independent of $|\Lambda|$ and such that the following estimates hold true
\begin{equation}\label{fluctuation2}\frac{\big|\det\Big{(}I+\frac{g}{\lambda}\frac{\partial V_{\mathrm{eff}}(n_i,{\mathcal{T}}^{-1}U^{\alpha})}{\partial U(n_j)}\Big{)}\big|}{\big|\det\Big{(}I+\frac{g}{\lambda}\frac{\partial V_{\mathrm{eff}}(n_i,{\mathcal{T}}^{-1}U)}{\partial U(n_j)}\Big{)}\big|}\geq D(\alpha).
 \end{equation} 
 \begin{equation}\label{fluctuation1} \int^{\infty}_{-\infty}D(\alpha)\rho\left(\alpha-\frac{g}{\lambda}V_{\mathrm{eff}}(n_0,{\mathcal{T}}^{-1}U^{\alpha})\right)\prod_{n\in |\Lambda|\setminus\{n_0\}}\frac{\rho\left(U^{\alpha}(n)-\frac{g}{\lambda}V_{\mathrm{eff}}(n,{\mathcal{T}}^{-1}U^{\alpha})\right)}{\rho\left(U(n)-\frac{g}{\lambda}V_{\mathrm{eff}}(n,{\mathcal{T}}^{-1}U)\right)}  \,d\alpha\geq C_{\mathrm{fluct}}(U(n_0)).
     \end{equation}     
     
The bounds (\ref{fluctuation1}) and (\ref{fluctuation2}) readily imply that, setting $U(n_0)=u$
\begin{equation}\label{conditionalbound}{\rho}^{\Lambda}_{n_0}(u)\leq \frac{\rho\left(u-\frac{g}{\lambda}V_{\mathrm{eff}}(n_0,{\mathcal{T}}^{-1}U)\right)}{C_{\mathrm{fluct}}(u)}.
\end{equation}
Lemma \ref{bdddensity} will follow from a precise control of the right-hand side of equation (\ref{conditionalbound}).
     
\par

We now execute the strategy which was outlined above. The ratio of determinants can be controlled through the following bound, where $\|M\|_1$ denotes the trace norm of a matrix $M$.

\begin{lem}\label{detabsbd} Let $A,B$ be square matrices with $I+B$ invertible.
 Then, \begin{equation}\left|\frac{\det\left(I+A\right)}{\det\left(I+B\right)}\right|\leq e^{\|(A-B)(I+B)^{-1}\|_1}.\end{equation}
\end{lem}

\begin{proof}
We make use of the elementary identities
\begin{equation}\frac{\det(I+A)}{\det(I+B)}=\det(I+A)(I+B)^{-1}
\end{equation}
 and
\begin{equation}(I+A)(1+B)^{-1}=I+(A-B)(I+B)^{-1}.
 \end{equation}
The proof is now finished making use of the inequality
\begin{equation*}\left|\det(1+M)\right|\leq e^{\|M\|_{1}}
    \end{equation*}
    which holds on the general setting of trace class operators, see \cite[Lemma 3.3]{trace}
\end{proof}
The triangle inequality for the trace norm implies the following.
\begin{cor}\label{ratiolbound}Under the above conditions
\begin{equation}
\Big|\frac{\det(I+B)}{\det(I+A)}\Big|\geq e^{-\sum_{m,n}|\left((A-B)(I+B)^{-1}\right)_{mn}|}\end{equation}
\end{cor}
Letting $A=\frac{g}{\lambda}\left(\frac{\partial V(n_i,\omega)}{\partial U(n_j)}\right)_{|\Lambda|\times |\Lambda|}$ and $B=\frac{g}{\lambda}\left(\frac{\partial V(n_i,\omega_{\alpha})}{\partial U(n_j)}\right)_{|\Lambda|\times |\Lambda|}$ and using lemma \ref{decay1} we see that, for $|g|<\lambda C_1^{-1}$, $(1+B)^{-1}$ has uniformly bounded operator norm. Using lemma \ref{bddtrace} and corollary \ref{ratiolbound} we conclude that (\ref{fluctuation2}) holds with $D(\alpha)=e^{-|g|^2C_3\left(|\alpha-U(n_0)|+C_4\right)}$.

 We now check that equation (\ref{fluctuation1}) holds when $\rho$ satisfies the fluctuation bound  (\ref{flucassump}). We divide the proof in two cases:
 \begin{enumerate}[label=(\Roman*)]
 \item{ Suppose that $c_2(\rho)>0$.\par
 Let $c_{\rho}=\max\{c_1(\rho),c_2(\rho)\}$.
 The left-hand side of (\ref{fluctuation1}) is bounded from below by
 \begin{equation}\int^{\infty}_{-\infty}D(\alpha)\rho\left(\alpha-\frac{g}{\lambda}V_{\mathrm{eff}}(n_0,{\mathcal{T}}^{-1}U^{\alpha})\right)\prod_{n\in \Lambda \setminus\{n_0\}}e^{-c_{\rho}|\omega(n)-\omega_{\alpha}(n)|(1+|\omega(n)|+|\omega_{\alpha}(n)|)}\,d\alpha
 \end{equation}
which equals
\begin{equation}\label{densityint}
\int^{\infty}_{-\infty}D(\alpha)\rho\left(\alpha-\frac{g}{\lambda}V_{\mathrm{eff}}(n_0,{\mathcal{T}}^{-1}U^{\alpha})\right) e^{\sum_{n\in \Lambda \setminus\{n_0\}}-c_{\rho}|\omega(n)-\omega_{\alpha}(n)|(1+|\omega(n)|+|\omega_{\alpha}(n)|\,)}\,d\alpha.
\end{equation}

Due to the triangle inequality and lemmas \ref{resamp} and \ref{improvedct}, we conclude that there is a positive constant 
$C=C(d,\|F\|_{\infty},g,\eta,\nu)$ with $\lim_{g\to 0}C(d,\|F\|_{\infty},g,\eta,\nu)<\infty$ such that for $n\neq n_0$
\begin{align*}
-c_{\rho}|\omega(n)-\omega_{\alpha}(n)|\left(1+|\omega(n)|+|\omega_{\alpha}(n)|\right)&\geq -c_{\rho}|\omega(n)-\omega_{\alpha}(n)|\left(1+2|\omega(n)|+|\omega_{\alpha}(n)-\omega(n)|\right)\\
&\geq  -\frac{|g|}{\lambda}c_{\rho}e^{-\nu|n-n_0|}\left(C^{2}|\alpha-U(n_0)|^2+2C|\alpha-U(n_0)| \right).\\
\end{align*}

Therefore,
\begin{align*}-c_{\rho}\sum_{n\in \Lambda\setminus\{n_0\}}|\omega(n)-\omega_{\alpha}(n)|\left(1+|\omega(n)|+|\omega_{\alpha}(n)|\right)&\geq -\frac{|g|}{\lambda}\frac{2c_{\rho}}{\left(1-e^{-\nu}\right)^d}\left(C^{2}|\alpha-U(n_0)|^2+2C|\alpha-U(n_0)| \right)).\\
\end{align*}}
  Thus, by choosing $\frac{|g|}{\lambda}$ sufficiently small so that $\frac{|g|}{\lambda}\frac{4c_{\rho}}{\left(1-e^{-\nu}\right)^d}C^{2}<{\overline{c}}_{\rho}$ and
 using the assumption \ref{fluctintegral1} we obtain
\begin{equation}
0<\int^{\infty}_{-\infty}D(\alpha)\frac{\rho\left(\alpha-\frac{g}{\lambda}V_{\mathrm{eff}}(n_0,{\mathcal{T}}^{-1}U^{\alpha})\right)}{\rho\left(U(n_0)-\frac{g}{\lambda}V_{\mathrm{eff}}(n_0,{\mathcal{T}}^{-1}U)\right)} \prod_{n\in \Lambda \setminus\{n_0\}}e^{-c_{\rho}|\omega(n)-\omega_{\alpha}(n)|(1+|\omega(n)|+|\omega_{\alpha}(n)|)}\,d\alpha<\infty.
\end{equation} this, together with (\ref{conditionalbound}),
verifies lemma \ref{bdddensity} when $c_2(\rho)>0$. If $\rho$ satisfies the assumption \ref{momentassumption} the above argument yields ${\rho}^{\Lambda}_{n_0}(u)\in L^{1}\left(\mathbb{R},|x|^{\varepsilon}dx\right)$.
\item{ Assume that $c_2(\rho)=0$:\par
Similarly to the above argument, the left-hand side of (\ref{fluctuation1}) is bounded from below by
\begin{equation}
\int^{\infty}_{-\infty}D(\alpha)\frac{\rho\left(\alpha-\frac{g}{\lambda}V_{\mathrm{eff}}(n_0,{\mathcal{T}}^{-1}U^{\alpha})\right)}{\rho\left(U(n_0)-\frac{g}{\lambda}V_{\mathrm{eff}}(n_0,{\mathcal{T}}^{-1}U)\right)}e^{-\frac{|g|}{\lambda}\frac{c_{1}({\rho})}{\left(1-e^{-\nu}\right)^d}C|\alpha-U(n_0)|}.
\end{equation}
Where, from (\ref{flucassump})
\begin{equation}
\frac{\rho\left(\alpha-\frac{g}{\lambda}V_{\mathrm{eff}}(n_0,{\mathcal{T}}^{-1}U^{\alpha})\right)}{\rho\left(U(n_0)-\frac{g}{\lambda}V_{\mathrm{eff}}(n_0,{\mathcal{T}}^{-1}U)\right)}\geq e^{-c_1(\rho)\left(|\alpha-U(n_0)|+2\frac{|g|}{\lambda}\right)}.
\end{equation}
Again, choosing $|g|$ sufficiently small we conclude that
\begin{equation}
0<\int^{\infty}_{-\infty}D(\alpha)\frac{\rho\left(\alpha-\frac{g}{\lambda}V_{\mathrm{eff}}(n_0,{\mathcal{T}}^{-1}U^{\alpha})\right)}{\rho\left(U(n_0)-\frac{g}{\lambda}V_{\mathrm{eff}}(n_0,{\mathcal{T}}^{-1}U)\right)} \prod_{n\in \Lambda \setminus\{n_0\}}e^{-c_{1}({\rho})|\omega(n)-\omega_{\alpha}(n)|}\,d\alpha<\infty.
\end{equation}
finishing the proof.}
\end{enumerate}

\section{The Hartree approximation for the Hubbard model}\label{hubbardext}

Let us now adapt the techniques from the previous sections to the Hubbard model. Recall that $H_{\mathrm{Hub}}$ is defined as

$$\begin{pmatrix}
\,H_{\uparrow}(\omega) & 0 \\
0 & H_{\downarrow}(\omega)\,\\
\end{pmatrix}
:=
\begin{pmatrix}
\, H_0+\lambda V_{\omega}+gV_{\uparrow}(\omega) & 0 \\
0 & H_0+\lambda V_{\omega}+gV_{\downarrow}(\omega)\,\\
\end{pmatrix}
$$
acting on $\ell^2\left(\mathbb{Z}^d\right)\oplus \ell^2\left(\mathbb{Z}^d\right)$.
The operators $H_0$ and $V_{\omega}$ are defined as before, i.e; 
$H_0+\lambda V_{\omega}$ is the standard Anderson model acting on 
$\ell^2\left(\mathbb{Z}^d\right)$. The effective potentials are defined as \begin{equation}\label{potentialHubbard}
\begin{pmatrix}
\,V_{\uparrow}(\omega)(n) \\
V_{\downarrow}(\omega)(n)\,\\
\end{pmatrix}
=\begin{pmatrix}
\,\mel{n}{F(H_{\downarrow})}{n} \\
\mel{n}{F(H_{\uparrow})}{n}\,\\
\end{pmatrix}.
\end{equation}

Mathematically, the treatment of the above model is very similar to the the proof of theorem \ref{main} above, therefore some details are skipped and we just indicate the required modifications.
\subsection{Existence of the Effective potential}\label{existHubb}
 Let $\Phi(X,Y):\ell^{\infty}\left(\mathbb{Z}^d\right)\oplus \ell^{\infty}\left(\mathbb{Z}^d \right)\rightarrow \ell^{\infty}\left(\mathbb{Z}^d\right)\oplus \ell^{\infty}\left(\mathbb{Z}^d \right)$ be given by
 $$\Phi(X,Y)(m,n):=\left(\mel{n}{F(H_0+V_{\omega}+gY)}{n},\mel{m}{F(H_0+V_{\omega}+gX)}{m}\right).$$
 using proposition \ref{Contraction}, we immediately reach
 \begin{equation}
 \|\Phi(X_1,Y_1)-\Phi(X_2,Y_2)\|_{\ell^{\infty}\left(\mathbb{Z}^d\right)\oplus \ell^{\infty}\left(\mathbb{Z}^d \right)}\leq |g|\frac{72\sqrt{2}}{\eta\left(1-e^{\nu'-\nu}\right)^d}\|F\|_{\infty}\left(\|X_1-X_2\|_{\ell^{\infty}\left(\mathbb{Z}^d\right)}+\|Y_1-Y_2\|_{\ell^{\infty}\left(\mathbb{Z}^d\right)}\right).
  \end{equation}
Therefore, if $|g|\frac{72\sqrt{2}}{\eta\left(1-e^{\nu'-\nu}\right)^d}\|F\|_{\infty}<1$ we conclude $\Phi$ has a unique fixed point $V_{\mathrm{eff}}=\left(V_{\uparrow},V_{\downarrow}\right)$ belonging to $\ell^{\infty}\left(\mathbb{Z}^d\right)\oplus \ell^{\infty}\left(\mathbb{Z}^d \right) $.

\subsection{Regularity of the effective potential}
Fix $\Lambda \subset \mathbb{Z}^d$ finite and define functions $\xi:\left(\ell^{\infty}\left(\Lambda\right)\oplus \ell^{\infty}\left(\Lambda \right)\right)\times \mathbb{R}^{n}\rightarrow \ell^{\infty}\left(\Lambda\right)\oplus \ell^{\infty}\left(\Lambda \right)$ through
\begin{equation}
\xi^{\uparrow}(V,\omega)(j)=V^{\uparrow}(j)-\mel{j}{F(H_0+\lambda\omega+gV_{\downarrow})}{j}.
\end{equation}
\begin{equation}
\xi^{\downarrow}(V,\omega)(j)=V^{\downarrow}(j)-\mel{j}{F(H_0+\lambda\omega+gV_{\uparrow})}{j}.
\end{equation}
 Our goal is to conclude $V^{\uparrow}$,$V^{\downarrow}$ are smooth functions of an arbitrary, but finite, list $(\omega(1),...,\omega(n))$. Again, this can be done via implicit function theorem once we check that the derivative
\begin{equation}
\frac{\partial \xi(\omega, V)(j)}{\partial V(l)}=
\delta_{jl}-\frac{\partial \mel{j}{F(H_0+\lambda\omega+gV)}{j}}{\partial V(l)}.
\end{equation} 
is non-singular.
Using lemma \ref{Contraction}, we have that for $\sharp\in \{\uparrow,\downarrow\}$
\begin{equation}
\Big|\frac{\partial \mel{j}{F(H_0+\lambda\omega+gV_{\sharp})}{j}}{\partial V(l)}\Big|\leq 
|g|\frac{72\sqrt{2}e^{-2\nu|j-l|}}{\eta}\|F\|_{\infty}.
\end{equation}

In particular, whenever $|g|\frac{144\sqrt{2}\|F\|_{\infty}}{\eta(1-e^{-2\nu})^{d}}<1$  we have that the operator $D\xi(\omega,.):\ell^{\infty}\left(\Lambda \right)\oplus \ell^{\infty}\left(\Lambda \right)\rightarrow \ell^{\infty}\left(\Lambda \right)\oplus \ell^{\infty}\left(\Lambda \right)$ has an inverse.
From the implicit function theorem it follows that $V$ is a smooth function of $(\omega(1),...,\omega(n))$ for $n=|\Lambda|$.
\subsection{Decay estimates}
The decay rate in the case of the Hubbard model is dictated by
\begin{equation}\label{Hubbarddecayup} \Big|\frac{\partial V_{\uparrow}(n)}{\partial \omega(m)}\Big |\leq 3|g|\|F\|_{\infty}\sum_{l}\tilde{K_{\downarrow}}(l,m)\Big|\frac{\partial V_{\downarrow}(l)}{\partial \omega(m)}\Big|+
\tilde{r_{\downarrow}}(n).
\end{equation}

\begin{equation}\label{Hubbarddecaydown} \Big|\frac{\partial V_{\downarrow}(n)}{\partial \omega(m)}\Big |\leq 3|g|\|F\|_{\infty}\sum_{l}\tilde{K_{\uparrow}}(l,m)\Big|\frac{\partial V_{\uparrow}(l)}{\partial \omega(m)}\Big|+
\tilde{r_{\uparrow}}(n).
\end{equation}

where, for $\sharp\in\{\uparrow,\downarrow\}$

  $$\tilde{G_{\sharp}}(l,m):=G_{\sharp}(l,n;t+i\eta)G_{\sharp}(n,l;t+i\eta)-G_{\sharp}(l,n;t-i\eta)G_{\sharp}(n,l;t-i\eta).$$
 $$r_{\sharp}(m,n):=G_{\sharp}(n,m;t+i\eta)G_{\sharp}(m,n;t+i\eta)-G_{\sharp}(n,m;t-i\eta)G_{\sharp}(m,n;t-i\eta).$$
 
$$\tilde{K_{\sharp}}(l,m):=\int^{\infty}_{-\infty}|\tilde{G_{\sharp}}(l,m)|\,dt.$$
$$\tilde{r_{\sharp}}(n):=\int^{\infty}_{-\infty} |r_{\sharp}(n)|\,dt.
$$
 
In particular,
\begin{equation}\label{Hubbardfulldecay}
\Big|\frac{\partial V_{\uparrow}(n)}{\partial \omega(m)}\Big|+\Big|\frac{\partial V_{\downarrow}(n)}{\partial \omega(m)}\Big |\leq 3|g|\|F\|_{\infty}\sum_{l}\left(\tilde{K_{\uparrow}}(l,m)+\tilde{K_{\downarrow}}(l,m)\right)\left(\Big|\frac{\partial V_{\uparrow}(l)}{\partial \omega(m)}\Big|+\Big|\frac{\partial V_{\downarrow}(l)}{\partial \omega(m)}\Big|\right)
+\left(\tilde{r_{\uparrow}}(n,m)+\tilde{r_{\downarrow}}(n,m)\right).
\end{equation}

The analysis from the previous sections applies and we obtain lemmas \ref{decay1},\ref{resamp},\ref{bddtrace} and \ref{improvedct} with $|.|$ being replaced by the matrix norm  $|M|=|M_{11}|+|M_{21}|$ for $M=
\begin{pmatrix}
M_{11}\\
M_{21}\,\\
\end{pmatrix}$.
The effective potential and its derivatives are to be interpreted as follows:

 \begin{equation*} V_{\mathrm{eff}}(n)=
\begin{pmatrix}
V^{\uparrow}_{\mathrm{eff}}(n)\\
V^{\downarrow}_{\mathrm{eff}}(n)\,\\
\end{pmatrix},
\,\,\,\,\,\,
\frac{\partial V_{\mathrm{eff}}(n)}{\partial \omega(m)}=
\begin{pmatrix}
\frac{\partial V^{\uparrow}_{\mathrm{eff}}(n)}{\partial \omega(m)}\\
\frac{\partial V^{\downarrow}_{\mathrm{eff}}(n)}{\partial \omega(m)}\,\\
\end{pmatrix}
\,\,\,\mathrm{and}\,\,\,
\frac{\partial^2 V_{\mathrm{eff}}(n)}{\partial \omega(m)\partial \omega(l)}=
\begin{pmatrix}
\frac{\partial^2 V^{\uparrow}_{\mathrm{eff}}(n)}{\partial \omega(m)\omega(l)}\\
\frac{\partial^2 V^{\downarrow}_{\mathrm{eff}}(n)}{\partial \omega(m)\omega(l)}\,\\
\end{pmatrix}.
\end{equation*}

\section{One dimensional Aspects:proof of theorem \ref{1dloc}}\label{details1d}
In this section we will prove theorem \ref{1dloc}. We let $H_{+}=H_{[0,\infty)\cap\mathbb{Z}}$ and denote by $G^{+}(m,n;z)$ the Green's function of $H^{+}$. Recall the definition of the Lyapunov exponent: \begin{equation}
\mathcal{L}(z)=-\mathbb{E}\left(\ln|G^{+}(0,0;z)|\right)
\end{equation}

\begin{equation}
\mathcal{L}_{\mathrm{And}}(z)=-\mathbb{E}\left(\ln|G^{+}_{\mathrm{And}}(0,0;z)|\right).
\end{equation}
Recall in this case $H_0=-\Delta$ hence, we define $H_{\mathrm{Hub}}$ acting on $\left(\ell^{2}\left(\mathbb{Z}\right)\oplus \ell^{2}\left(\mathbb{Z} \right)\right)$ by
\begin{equation} H_{\mathrm{Hub}}=\begin{pmatrix}
\,H_{\uparrow}(\omega) & 0 \\
0 & H_{\downarrow}(\omega)\,\\
\end{pmatrix}
\end{equation}
where, denoting by $H_{\mathrm{And}}$ the standard Anderson model $-\Delta+V_{\omega}$ on $\ell^2\left( \mathbb{Z}\right)$,
\begin{equation}\label{Hubbarddefonedim}\begin{pmatrix}
\,H_{\uparrow}(\omega) & 0 \\
0 & H_{\downarrow}(\omega)\,\\
\end{pmatrix}
:=
\begin{pmatrix}
\,H_{\mathrm{And}}+gV_{\uparrow}(\omega) & 0 \\
0 & H_{\mathrm{And}}+gV_{\downarrow}(\omega)\,\\
\end{pmatrix}.
\end{equation}
The effective potentials are defined as 
(\ref{eff}). In the theorem below, we will use an abbreviation and $\mathcal{L}(z)$ will refer to the Lyapunov exponent of either $H_{\uparrow}$ or $H_{\downarrow}$ whereas $\mathcal{L}_{\mathrm{And}}(z)$ will denote the Lyapunov exponent of the Anderson model on $\ell^2\left( \mathbb{Z}\right)$.

\subsection{Proof of Lemma \ref{Fekete}}
For simplicity we set $C=0$. The general statement will follow by considering the related sequence  $b_n:=a_n+C$. Given integers $L$ and $\ell$ with $L>>\ell$, our goal is to bound $\frac{a_{L}}{L}$ from above in terms of $\frac{a_{\ell}}{\ell}$. As a initial step, observe that by (\ref{subad}) we have
 \begin{equation*}
a_L\leq a_{L-\delta \log L-\ell}+a_{\ell}.
\end{equation*}
Iterating the above procedure $k+1$ times for \begin{equation}k=k_{\ell,L}:=\lfloor \frac{L-2\ell-\delta \log L}{\delta \log L+\ell}\rfloor
\end{equation}
we obtain
\begin{equation*}
a_L\leq \left(k+2\right) a_{\ell}
\end{equation*}
In the above iteration we have made use of the fact that in the assumption (\ref{subad}) the remainder $r$ can be adjusted as long as it satisfies the inequality given there. Thus,
\begin{equation}\label{startscal}
\frac{a_L}{L}\leq \frac{\ell(k+2)}{L}\frac{a_{\ell}}{\ell}
\end{equation}
Before proceeding with the proof, a few remarks are in order. Firstly, nothing is achieved by holding $\ell$ fixed and letting $L\to \infty$ directly on equation (\ref{startscal}) since this only yields the upper bound of zero. A second attempt would be showing that letting $\ell \to \infty$ (hence, $L \to \infty$ as well) implies that the ratio $\frac{k\ell}{L}$ converges to one. However, as
\begin{equation}
q_{\ell,L}-\frac{\ell}{L}\leq \frac{k\ell}{L}\leq q_{\ell,L}
\end{equation} 
for the choice
\begin{equation}\label{qdef}
q_{\ell,L}=\frac{1-2\frac{\ell}{L}-\delta \frac{\log L}{L}}{1+\frac{\delta \log L}{\ell}}
\end{equation}
we see that $\frac{k\ell}{L}$ converges to one as $\ell \to \infty$ only along a subsequence where 
\begin{equation}\label{initialreq}
\frac{\ell}{L}\to 0\,\,\,\,\mathrm{and}\,\,\,\, \frac{\log L}{\ell}\to 0.
\end{equation}
Taking this into account, let $\varepsilon>0$ be given and $\ell_1$ be an initial scale to be determined. Let $L>>\ell_1$ be a positive integer to be interpreted as a larger scale. Iterating equation (\ref{startscal}) throughout a sequence of scales \begin{equation}\ell_1<\ell_2<...<\ell_{N_L} \leq L<\ell_{N_L+1}<...
 \end{equation} satisfying, for some $p>0$,
\begin{equation}
p\log \left(\ell_j\right)\leq \log \ell_{j+1} < p^2\log \left(\ell_j\right).
\end{equation}
and
\begin{equation}\label{summability}
\sum^{\infty}_{j=1}\frac{\ell_{j}}{\ell_{j+1}}<\infty
\end{equation}

we reach, for $q_{\ell,L}$ defined in (\ref{qdef}),
\begin{equation}
\frac{a_L}{L}\leq \left(q_{\ell_{N_L},L}+\frac{2\ell_{N_L}}{L}\right)\prod^{N_{L}-1}_{j=1} \left(q_{\ell_{j},\ell_{j+1}}+\frac{2\ell_j}{\ell_{j+1}}\right)
\frac{a_{\ell_1}}{\ell_1}.
\end{equation}
Since $q_{\ell_{j},\ell_{j+1}}\to \infty$ as $j\to \infty$ Due to (\ref{summability}), we conclude that the value of $\ell_1$ can be chosen(independently of $L$) so that 
 \begin{equation}\label{tails}
 \sum^{N_{L}-1}_{j=1} \log\left(q_{\ell_{j},\ell_{j+1}}+\frac{2\ell_j}{\ell_{j+1}}\right)+\log\left(q_{\ell_{N_L},L}+\frac{2\ell_{N_L}}{L}\right)<\varepsilon.
\end{equation}
Thus

\begin{equation}\label{conclusionseq}
\frac{a_L}{L}\leq e^{\varepsilon}\frac{a_{\ell_1}}{\ell_1}.
\end{equation}
Moreover, the above conclusion holds for any integer $\ell_1$ sufficiently large, as long and $L>>\ell_1$. In particular, we can also require that 
\begin{equation}\label{infimumseq}
\frac{a_{\ell_1}}{\ell_1}\leq \inf_{n}\frac{a_n}{n}+\varepsilon.
\end{equation}
Combining equations (\ref{conclusionseq}) and (\ref{infimumseq}) the proof is finished letting 
$\varepsilon \to 0$.
\subsection{Proof of lemma \ref{mixinglem}}

We will show that the following inequality holds
\begin{equation}\label{upperdecoupling}
\mathbb{E}_{[0,n+m+r]}\left(|\hat G(0,n+m+r;z)|^s \right)\leq  C_{\mathrm{AP}}e^{C(\eta,g,\|F\|_{\infty}) e^{-\nu'r}(m^2+n^2)}\mathbb{E}_{[0,n]}\left(|\hat G(0,n;z)|^s \right) \mathbb{E}_{[0,m]}\left(|\hat G(0,m;z)|^s \right)
\end{equation}
where $\mathbb{E}_{[0,n]}$ denotes the expectation with respect to the variables $U(0),...,U(n)$ and $C_{\mathrm{AP}}$ is, up to a multiplication by a constant independent of $m,n$ and $r$, the constant obtained on the \emph{a priori} from lemma \ref{apriori}.
Let us change variables according to
\begin{equation}\label{change1d}
\left(\omega(1),...,\omega(n+1),\omega(n+r),...,\omega(n+r+m)\right)\mapsto 
\left(U(1),..., U(n+1), U(n+r),...,U(n+r+m)\right).
\end{equation}
We remark that the variables $\omega(n+2),...,\omega(n+r-1)$ are fixed in this process.

Note that by lemma \ref{apriori} and a geometric resolvent expansion we have

\begin{equation}\label{firstapriori}
\mathbb{E}\left(\hat G(0,n+m+r;z)|^s \right)\leq 
C_{\mathrm{AP}}\mathbb{E}_{\neq n+1,n+r}\left(|\hat G(0,n;z)|^s |\hat G(n+r+1,n+r+m;z)|^s \right).
\end{equation}
where $\mathbb{E}_{\neq n+1,n+r}$ indicates that the variables $\omega(n+1)$ and $\omega(n+r)$ were integrated out. 
Observe that the corresponding Jacobian has the structure
\begin{equation*}\label{jacobian}
\mathcal{J}=\begin{pmatrix}
\mathcal{A}_{n\times n} & \mathcal{B}_{n\times(m+r)}\\
\mathcal{C}_{(m+r) \times n} & \mathcal{D}_{(m+r)\times(m+r)}\\
\end{pmatrix}
\end{equation*}
where
\begin{equation*}\mathcal{A}_{jk}=\delta_{jk}+g\frac{\partial V_{\mathrm{eff}}(j)}{\partial U(k)},\,\,\, \mathcal{B}_{jk}=g\frac{\partial V_{\mathrm{eff}}(j)}{\partial U(n+k)},\,\,\, \mathcal{C}_{jk}=g\frac{\partial V_{\mathrm{eff}}(n+j)}{\partial U(k)}.
\end{equation*}
Moreover,
\begin{equation*}\mathcal{D}=\begin{pmatrix}
\mathcal{I}_{r\times r} & \mathcal{Q}_{r\times m}\\
\mathcal{O}_{m \times r} & \mathcal{P}_{m\times m}\\
\end{pmatrix}
\end{equation*}
where $\mathcal{I}_{r\times r}$ is the identity matrix and
\begin{equation*}\mathcal{P}_{jk}=\delta_{jk}+g\frac{\partial V_{\mathrm{eff}}(n+r+j)}{\partial U(n+r+k)},\,\,\, \mathcal{Q}_{jk}=g\frac{\partial V_{\mathrm{eff}}(n+j)}{\partial U(n+r+k)},\,\,\, \mathcal{O}_{jk}=g\frac{\partial V_{\mathrm{eff}}(n+r+j)}{\partial U(n+k)}.
\end{equation*}

By the Schur complement formula  
\begin{equation}\label{determinantschur}\det \mathcal{J}=\det \mathcal{A}\det \mathcal{P} \det\left(\mathcal{I}_{{m+r}\times{m+r}}-{\mathcal{D}}^{-1}
\mathcal{C}{\mathcal{A}}^{-1}\mathcal{B}\right)\det\left(\mathcal{I}_{{m}\times{m}}-{\mathcal{P}}^{-1}
\mathcal{O}\mathcal{Q}\right).
\end{equation}
  where, according to the estimate \ref{lemmapotential1}, the matrices $\mathcal{B}$ and $\mathcal{Q}$ have entries which decay exponentially away from their lower-left corner. Likewise, the entries of $\mathcal{C}$, $\mathcal{O}$ decay exponentially away from their upper-right corner. It readily follows from lemma that \ref{detabsbd}, 
  \begin{eqnarray*}
  \det\left(\mathcal{I}_{{m+r}\times{m+r}}-{\mathcal{D}}^{-1}
\mathcal{C}{\mathcal{A}}^{-1}\mathcal{B}\right)\det\left(\mathcal{I}_{{m}\times{m}}-{\mathcal{P}}^{-1}
\mathcal{O}\mathcal{Q}\right)\leq C(\eta,\nu,g,\|F\|_{\infty}).
  \end{eqnarray*} 
  therefore, for $C$ as above,
  \begin{equation*}
      \det \mathcal{J}\leq C\det \mathcal{A}\det \mathcal{P}.
  \end{equation*}
  Let $\rho(l)=\rho\left(U(l)-gV_{\mathrm{eff}}(l)\right)$. We obtain a decoupling estimate by observing that, setting $U(l)=0$ for $l\geq n+r+1$ would only alter $$V_{\mathrm{eff}}(j)\,\,\mathrm{and}\,\,\frac{\partial V_{\mathrm{eff}}(j)}{\partial U(k)} $$ by at most a factor which decays as $C(\eta,g,\|F\|_{\infty})e^{-\nu(r+|j-k|)}$ for $1\leq j,k\leq n$. This follows from the exponential decay on lemmas \ref{decay1}, \ref{improvedct} and \ref{localbeh}. Similarly, we can set $U(l)=0$ for $l\leq n$ and this only changes $$V_{\mathrm{eff}}(j)\,\,\mathrm{and}\,\,\frac{\partial V_{\mathrm{eff}}(j)}{\partial U(k)} $$ by at most a factor bounded which decays as $C(\eta,\lambda,g,\|F\|_{\infty})e^{-\nu(r+|j-k|)}$ for $n+r\leq j,k\leq m+n+r$.
The above process yields two independent measures 
$${\det}_{0}\left (I+g\frac{\partial V_{\mathrm{eff}}(j)}{\partial U(k)}\right)_{ [0,n]} \prod_{0\leq l\leq n}\rho^{0}(l)\,dU(l)$$ and 
$${\det}_{0}\left ( I+g\frac{\partial V_{\mathrm{eff}}(j)}{\partial U(k)}\right)_{ [n+r+1,n+r+m]}\prod_{n+r+1\leq n+r+m}\rho^{0}(l)\,dU(l).$$
  
   Using lemma \ref{detabsbd}
  we then arrive at an inequality of the type
\begin{align*}\mathbb{E}_{[0,m+n+r]}&\left(|\hat G(0,n+m+r;z)|^s \right)\\
&\leq C_{\mathrm{AP}}e^{ C(\eta,g,\|F\|_{\infty})e^{-\nu r}(m^2+n^2)}\int |\hat G(0,n;z)|^s {\det}_{0}\left (I+g\frac{\partial V_{\mathrm{eff}}(j)}{\partial U(k)}\right)_{ [0,n]} \prod_{0\leq l\leq n}\rho^{0}(l)\,dU(l)\\
&\times | G_{+}(n+r,n+r+m;z)|^s {\det}_{0}\left ( I+g\frac{\partial V_{\mathrm{eff}}(j)}{\partial U(k)}\right)_{ [n+r+1,n+r+m]}\prod_{n+r+1\leq n+r+m}\rho^{0}(p)\,dU(p).\\
\end{align*}
 Rewriting the above conclusion in terms of expectations  we obtain
\begin{equation}
\mathbb{E}_{[0,m+n+r]}\left(|\hat G(0,n+m+r;z)|^s \right)\leq C_{\mathrm{AP}}e^{ C(\eta,g,\|F\|_{\infty})e^{-\nu r}(m^2+n^2)} \mathbb{E}_{[0,n]}\left(|\hat G(0,n;z)|^s\right)|  \mathbb{E}_{[n+r+1,n+m+r]}\left(G_{+}(n+r,n+r+m;z)|^s\right).
\end{equation}
At the expense of increasing $C_{\mathrm{AP}}$ the above expectations can be replaced by expectations over the full probability space, which yields the desired conclusion by translation invariance.\qedhere

\subsection{Proof of (\ref{upperbound})}
\label{detailslowerbound}
We shall modify the proof of lemma \ref{mixinglem} to obtain an ``super-additive" estimate of the form
\begin{equation}\label{goallowerdec}
\mathbb{E}_{[0,n+m+r]}\left(|\hat G(0,n+m+r;z)|^s \right)\geq  \underline{C}(s,E)e^{\varphi(s,z)r}e^{-C(\eta,g,\|F\|_{\infty}) e^{-\nu'r}(m^2+n^2)}\mathbb{E}_{[0,n]}\left(|\hat G(0,n;z)|^s \right) \mathbb{E}_{[0,m]}\left(|\hat G(0,m;z)|^s \right)
\end{equation}
where the constant $\underline{C}(s,z)$ can be chosen locally uniform in $z$ and $s\in (0,1)$. Since the argument is very similar to the one in the proof of (\ref{upperdecoupling}), we only explain the key modification which consists in obtaining a lower bound for $\mathbb{E}_{[n+1,n+r]}\left(|\hat G(n+1,n+r;z)|^s \right)$ as follows. We start by writing
\begin{equation}\label{bottleneck}|\hat G(n+1,n+r;z)|^s=|\hat G(n+1,n+1;z)|^s|\hat G(n+2,n+r;z)|^s.
\end{equation}
Using Jensen's inequality we have that, for any $\varepsilon\in (0,s)$,
\begin{equation}\label{jensen}
\mathbb{E}_{n+1}\left(|\hat G(n+1,n+1;z)|^s\right)\geq \mathbb{E}_{n+1}\left(|\hat G(n+1,n+1;z)|^{-\varepsilon}\right)^{\frac{-s}{\varepsilon}}.
\end{equation}
where $\mathbb{E}_{n+1}$ denotes the conditional expectation with respect to $U(n+1)=\omega(n+1)+gV_{\mathrm{eff}}(n+1)$. Making use of the discrete Riccati equation \cite[Proposition 12.1]{A-W-B} we obtain
\begin{equation}\label{riccati}
\mathbb{E}_{n+1}\left(|\hat G(n+1,n+1;z)|^{-\varepsilon}\right)=\mathbb{E}_{n+1}\left(|U(n+1)-z-\hat G(n+2,n+2;z)|^{\varepsilon}\right).
\end{equation}
Equations (\ref{bottleneck}), (\ref{jensen}) and (\ref{riccati}) together with lemma \ref{bdddensity} yield a lower bound
\begin{equation}
\mathbb{E}_{[n+1,n+r]}\left(|\hat G(n+1,n+r;z)|^s\right)\geq \underline{C}(z,s)\mathbb{E}_{[n+2,n+r]}\left(|\hat G(n+2,n+r;z)|^s\right)
\end{equation}
which in combination with (\ref{lowerbound}) implies that, after a suitable adjustment of the constant $\underline{C}(s,z)$,
\begin{equation}
\mathbb{E}_{[n+1,n+r]}\left(|\hat G(n+1,n+r;z)|^s\right)\geq \underline{C}(z,s)e^{\varphi(s,z)r}.
\end{equation}
Equation (\ref{goallowerdec}) follows from the above inequality combined with a decoupling estimate analogous to the one in the proof of (\ref{upperdecoupling}). Again, choosing $r$ comparable to $\max\{\log m,\log n\}$ we obtain
\begin{equation}
\mathbb{E}_{[0,n+m+r]}\left(|\hat G(0,n+m+r;z)|^s \right)\geq  \underline{C}(s,z)e^{\varphi(s,z)r}\mathbb{E}_{[0,n]}\left(|\hat G(0,n;z)|^s \right) \mathbb{E}_{[0,m]}\left(|\hat G(0,m;z)|^s \right)
\end{equation}
multiplying both sides of the above inequality by $e^{-\varphi(s,z)(m+n)}$ and taking logarithms we conclude that the sequence $b_n=\log e^{-\varphi(s,z)n}\left( \mathbb{E}_{[0,n]}\left(|\hat G(0,n;z)|^s \right)\right)$ satisfies
\begin{equation}
b_{n+m+r}\geq \log(\underline{C}(s,z))+b_n+b_m.
\end{equation}
 The bound (\ref{upperbound}) now follows from an application of the supperaditive version of lemma \ref{Fekete}.

\section{H\"older Continuity for the integrated density of states at weak interaction}
\label{idssection}

In this section we shall address the problem of H\"older continuity for the integrated density of states for the Hubbard model with respect to energy, disorder and interaction. Our results follow from modifications of the methods in \cite{H-K-S} and references therein after we have established the existence of a suitable conditional density as in lemma \ref{bdddensity}.

Let's now prove theorem \ref{thmids}, starting from H\"older continuity with respect to energy, equation (\ref{IDSenergy}). We proceed as in \cite[Section 2]{H-K-S}. For simplicity, we replace $H_{\mathrm{Hub}}$ by $H$ defined in (\ref{toymodel}). The arguments given below will apply directly to $H_{\uparrow}$ and $H_{\downarrow}$ and, therefore, suffice to show the same result for $H_{\mathrm{Hub}}$.\par
 Fix an energy interval $I$ of length $\varepsilon>0$ centered at $E\in \mathbb{R}$. The idea is to use the H\"older continuity of $N_0$ and the resolvent identity to reach the following inequality for $\varepsilon<<1$ and $|I|=\varepsilon$, where we denote by $P_{\Lambda}(I)$ the spectral projection of $H^{\Lambda}$ on the interval $I$.
 \begin{equation}\label{goalids}(1-\mathrm{o}(\varepsilon))\mathbb{E}\left( \mathrm{Tr}P_{\Lambda}(I)\right)\leq C(I,\rho)\varepsilon^{\alpha}|\Lambda|.
\end{equation}
Dividing both sides of (\ref{goalids}) by $|\Lambda|$ and letting $|\Lambda|\to \infty$ gives (\ref{IDSenergy}). To obtain (\ref{goalids}) we fix an interval $J$ containing $I$ with $|J|$ to be determined. We then write, with $P_{0,\Lambda}(J)=P\left(H_{0}^{\Lambda}\right)(J)$,
\begin{equation}\label{eq0}
\mathrm{Tr}(P_{\Lambda}(I))=\mathrm{Tr}(P_{\Lambda}(I)P_{0,\Lambda}(J))+\mathrm{Tr}(P_{\Lambda}(I)P_{0,\Lambda}(J^c)).
\end{equation}
Note \begin{equation}\mathrm{Tr}(P_{\Lambda}(I)P_{0,\Lambda}(J))\leq \mathrm{Tr}(P_{0,\Lambda}(J)).
\end{equation}
The above inequality combined with to the H\"older continuity of $N_0$ with respect to $E\in \mathbb{R}$ 
\begin{equation}|N_0(E)-N_0(E')|\leq C(I,d)|E-E'|^{\alpha_0}.
\end{equation}
yields, for $|\Lambda|$ sufficiently large depending only on $J$,
\begin{equation}\label{1st}
\mathrm{Tr}(P_{\Lambda}(I)P_{0,\Lambda}(J))\leq C(J,d)|J|^{\alpha_0}|\Lambda|.
\end{equation}
We now estimate the second term on the left-hand side of equation (\ref{eq0}). By the resolvent identity,
\begin{equation}\label{remainingterm}
\mathrm{Tr}\left(P_{\Lambda}(I)P_{0,\Lambda}(J^c)\right)=\mathrm{Tr}\left(P_{\Lambda}(I)(H-E)P_{0,\Lambda}(J^c)(H_{0,\Lambda}-E)^{-1}\right)-\lambda \mathrm{Tr}\left(P_{\Lambda}(I)U^{\Lambda}P_{0,\Lambda}(J^c)(H_{0,\Lambda}-E)^{-1}\right).
\end{equation}
Where we have written $U=V_{\omega}+\frac{g}{\lambda}V_{\mathrm{eff}}$. Moreover,  using using functional calculus and that $E$ is the center of $I$, we estimate the first term on the left-hand side of equation (\ref{remainingterm}) by
\begin{equation}\label{2nd}
\mathrm{Tr}\left((P_{\Lambda}(I))(H^{\Lambda}-E)P_{0,\Lambda}(J^c)(H_{0,\Lambda}-E)^{-1}\right)\leq \frac{|I|}{|J|-|I|}\mathrm{Tr}(P_{\Lambda}(I)).
\end{equation}
Now, the second term in in equation (\ref{remainingterm}) can be controlled by means of
\begin{align*}
-\lambda \mathrm{Tr}\left(P_{\Lambda}(I)U^{\Lambda}P_{0,\Lambda}(J^c)(H_{0,\Lambda}-E)^{-1}\right)=&-\lambda \mathrm{Tr}\left((H^{\Lambda}-E)(P_{\Lambda}(I))U^{\Lambda}P_{0,\Lambda}(J^c)(H_{0,\Lambda}-E)^{-2}\right)\\
&+\lambda^2\mathrm{Tr}\left(U^{\Lambda}(P_{\Lambda}(I))U^{\Lambda}P_{0,\Lambda}(J^c)(H_{0,\Lambda}-E)^{-2}\right).\\
&=A+B
\end{align*}
Now, because $U^{\Lambda}$ is unbounded, we continue a slight modification of the argument in \cite{H-K-S}.
 The only difference is that we bound term (A) above (which corresponds to \cite[(iii) in equation (2.6)]{H-K-S} as
\begin{equation}\label{iii}
|\mathrm{Tr}\left((H^{\Lambda}-E)(P_{\Lambda}(I))U^{\Lambda}P_{0,\Lambda}(J^c)(H_{0,\Lambda}-E)^{-2}\right)|\leq \frac{|I|}{(|J|-|I|)^2}|\mathrm{Tr}\left(P_{\Lambda}(I)U^{\Lambda}\right)|.
\end{equation}
At this point, with an estimate analogous to the one in the proof of Proposition 3.2 in \cite{C-H-K} we reach \begin{equation}
\mathbb{E}\left(|\mathrm{Tr}P_{\Lambda}(I)V_{\omega}|\right)\leq \lambda^{-1} \sup_{m\in \mathbb{N}}\Big\{\int^{(m+1)\varepsilon}_{m\varepsilon}\omega_j\rho(\omega_j)\,d\omega_j\Big\}|\Lambda| \,\,\,\,\varepsilon=|I|.
\end{equation}
Thus, with $M_1(\varepsilon):=\sup_{m\in \mathbb{N}}\Big\{\int^{(m+1)\varepsilon}_{m\varepsilon}\omega_j\rho(\omega_j)\,d\omega_j\Big\}$,
\begin{equation}\label{3rd}
\lambda|\mathrm{Tr}(H_{\Lambda}-E)P_{\Lambda}(I)U^{\Lambda}P_{0,\Lambda}(J^c)(H^{\Lambda}_{0}-E)^{-2}|\leq \frac{\lambda|I|}{(|J|-|I|)^2}(\frac{M_1(\varepsilon)}{\lambda}+\frac{g\|F\|_{\infty}}{\lambda})|\Lambda|.
\end{equation}
Similarly, with $M_2(\varepsilon):=\sup_{m\in \mathbb{N}}\Big\{\int^{(m+1)\varepsilon}_{m\varepsilon}\omega^2_j\rho(\omega_j)\,d\omega_j\Big\}$, we estimate term (B) through
\begin{equation}\label{4th}
\lambda^2|\mathrm{Tr}U^{\Lambda}(P_{\Lambda}(I))U^{\Lambda}P_{0,\Lambda}(J^c)(H_{0,\Lambda}-E)^{-2}|\leq \frac{4\lambda^2}{(|J|-|I|)^2}\left(\frac{M_2(\varepsilon)}{\lambda}|\Lambda|+\frac{g^2}{\lambda^2}\mathrm{Tr}(P_{\Lambda}(I))\right).
\end{equation}
Due lemma \ref{bdddensity} and the Wegner estimate (see \cite[theorem 4.1]{A-W-B}) we conclude that
\begin{equation}\label{wegner}
\mathrm{Tr}(P_{\Lambda}(I))\leq \frac{C}{\lambda}|I||\Lambda|.
\end{equation}

 Choosing the interval $J$ such that $|J|=\varepsilon^{\delta}$ for $\delta<1$, keeping in mind the assumption $g^2<\lambda$ and combining the bounds (\ref{1st}), (\ref{2nd}), (\ref{3rd}), (\ref{4th}), (\ref{wegner}) and optimizing over $\delta$ gives $\delta=\frac{1}{2+\alpha_0}$ therefore we reach
(\ref{goalids}) for $\alpha\in [0,\frac{\alpha_0}{2+\alpha_0}]$ and (\ref{IDSenergy}) is proven.\par To show \ref{IDSdisorder} we follow the proof of theorem 1.2 in \cite{H-K-S}. We fix $\lambda,\lambda'\in J$ and $E\in I$. As explained in \cite{H-K-S}, using H\"older continuity with respect to energy given by equation (\ref{IDSenergy}), trace identities and ergodicity of $H_{\lambda,g}$ and $H_{\lambda',g'}$, it suffices to estimate $\mathbb{E}\left( \mathrm{Tr} P_0\varphi(H_{\lambda,g})(\varphi(H_{\lambda,g})-\varphi(H_{\lambda',g'}))P_0\right)$ where $\varphi$ is a smooth function such that

\begin{equation}
\left\{
\begin{array}{lll}
\varphi\equiv 1\; \mathrm{on} \;(-\infty,E], \\
\varphi\equiv 0 \;\mathrm{on} \; (-\infty,E+|\lambda-\lambda'|^{\delta}+|g-g'|^{\delta})^c , \\
\|\varphi^{(j)}\|_{\infty}\leq C\left(|\lambda-\lambda'|^{\delta}+|g-g'|^{\delta}\right)^{-j},\,j=1,2...,3d+4 \;\;\;
\end{array}
\right.
\end{equation}
with $\delta>0$ to be determined. The need for a high regularity of $\varphi$ is due to the fact that the random potential $V_{\omega}$ may be unbounded.
 Let $\tilde{\varphi}$ be an almost analytic extension of $\varphi$ of order $3+3d$. In particular, $\tilde{\varphi}$ is defined in a complex neighborhood of the support of $\varphi$ and if $z=E+i\eta$ we have that
\begin{equation}\label{order3}|\partial_{\overline{z}}\tilde{\varphi}(z)|\leq |\eta|^{3d+3}|\varphi^{(3d+4)}(E)|.
\end{equation}
 By the Helffer-Sj\"ostrand formula,
\begin{align*}
 \mathrm{Tr}\left( P_0\varphi(H_{\lambda,g})(\varphi(H_{\lambda,g})-\varphi(H_{\lambda',g'}))P_0\right)&=\frac{1}{\pi}\int_{\mathbb{C}}\partial_{\overline{z}}\tilde{\varphi}\,\mathrm{Tr}P_0\varphi(H_{\lambda,g})R_{\lambda,g}(z)\left(\lambda 'U_{\lambda',g'}-\lambda U_{\lambda,g}\right)R_{\lambda',g'}(z)P_0\,d^2z\\
&=\frac{(\lambda'-\lambda)}{\pi}\int_{\mathbb{C}}\partial_{\overline{z}}\tilde{\varphi}\,\mathrm{Tr}P_0\varphi(H_{\lambda,g})R_{\lambda,g}(z)V_{\omega}R_{\lambda',g'}(z)P_0\,d^2z\\
&+\frac{(g'-g)}{\pi}\int_{\mathbb{C}}\partial_{\overline{z}}\tilde{\varphi}\,\mathrm{Tr}P_0\varphi(H_{\lambda,g})R_{\lambda,g}(z)V_{\mathrm{eff},\lambda}(g)R_{\lambda',g'}(z)P_0\,d^2z\\
&+\frac{g'}{\pi}\int_{\mathbb{C}}\partial_{\overline{z}}\tilde{\varphi}\,\mathrm{Tr}P_0\varphi(H_{\lambda,g})R_{\lambda,g}(z)\left(V_{\mathrm{eff},\lambda'}(g')-V_{\mathrm{eff},\lambda}(g)\right)R_{\lambda',g'}(z)P_0\,d^2z.\\
\end{align*}
Since the last two terms enjoy a better modulus of H\"oldr continuity (since they do not involve $V_{\omega}$) and can be treated as in \cite{H-K-S}, we shall only estimate the first of the above integrals.
By the resolvent identity,
\begin{align*}R_{\lambda,g}(z)V_{\omega}R_{\lambda',g'}(z)=&R_{\lambda,g}(z)V_{\omega}R_{\lambda,g}(z)+(\lambda'-\lambda)R_{\lambda,g}(z)V_{\omega}R_{\lambda',g'}(z)V_{\omega}R_{\lambda,g}(z)+(g'-g)R_{\lambda,g}(z)V_{\omega}R_{\lambda',g'}(z)V_{\mathrm{eff},\lambda'}(g')R_{\lambda,g}(z)\\
&+gR_{\lambda,g}(z)V_{\omega}R_{\lambda',g'}(z)(V_{\mathrm{eff},\lambda'}(g')-V_{\mathrm{eff},\lambda}(g))R_{\lambda,g}(z).
\end{align*}

The above considerations lead to a perturbative expansion of $\frac{(\lambda'-\lambda)}{\pi}\int_{\mathbb{C}}\partial_{\overline{z}}\tilde{\varphi}\,\mathrm{Tr}P_0\varphi(H_{\lambda,g})R_{\lambda,g}(z)V_{\omega}R_{\lambda',g'}(z)P_0\,d^2z$ into four terms. We will show below that each of them  can be bounded in terms of powers of either $|\lambda-\lambda'|$ or $|g-g'|$.\par
We start by estimating $\mathbb{E}\left( \Big|\frac{(\lambda-\lambda')^2}{\pi}\int_{\mathbb{C}}\partial_{\overline{z}}\tilde{\varphi}\,\mathrm{Tr}P_0\varphi(H_{\lambda,g})R_{\lambda,g}(z)V_{\omega}R_{\lambda',g'}(z)V_{\omega}R_{\lambda,g}(z)P_0\,d^2z\Big|\right)$ with a slight modification of equation (3.15) in \cite{H-K-S} since  $V_{\omega}$ is unbounded. By the Combes-Thomas bound, equation (\ref{order3}) and the choice of $\varphi$
\begin{eqnarray*}
\mathbb{E}\left( \Big|\frac{(\lambda-\lambda')^2}{\pi}\int_{\mathbb{C}}\partial_{\overline{z}}\tilde{\varphi}\,\mathrm{Tr}P_0\varphi(H_{\lambda,g})R_{\lambda,g}(z)V_{\omega}R_{\lambda',g'}(z)V_{\omega}R_{\lambda,g}(z)P_0\,d^2z\Big|\right)\leq C(d)\left(1+\mathbb{E}^2(|V_{\omega}|)\right)\frac{|\lambda-\lambda'|^{2}}{(|\lambda-\lambda'|^{\delta}+|g-g'|^{\delta})^{3d+4}}.
\end{eqnarray*}
Similarly,
\begin{align*}
\mathbb{E}&\left( \Big|\frac{(\lambda-\lambda')(g-g')}{\pi}\int_{\mathbb{C}}\partial_{\overline{z}}\tilde{\varphi}\,\mathrm{Tr}P_0\varphi(H_{\lambda,g})R_{\lambda,g}(z)V_{\mathrm{eff},\lambda'}(g')R_{\lambda',g'}(z)V_{\omega}R_{\lambda,g}(z)P_0\,d^2z\Big|\right)\\
&\leq C(d)\left(1+\mathbb{E}(|V_{\omega}|)\right)\frac{|\lambda-\lambda'||g-g'|}{(|\lambda-\lambda'|^{\delta}+|g-g'|^{\delta})^{3d+4}}.
\end{align*} Moreover, using lemma \ref{comparepotentials} with the the explicit dependence on $\omega$ given there, we obtain
\begin{align*}
\mathbb{E}&\left( \Big|\frac{g(\lambda-\lambda')}{\pi}\int_{\mathbb{C}}\partial_{\overline{z}}\tilde{\varphi}\,\mathrm{Tr}P_0\varphi(H_{\lambda,g})R_{\lambda,g}(z)V_{\omega}R_{\lambda',g'}(z)(V_{\mathrm{eff},\lambda'}(g')-V_{\mathrm{eff},\lambda}(g))R_{\lambda,g}(z)P_0\,d^2z\Big|\right)\leq\\
 &C(d)\left(1+\mathbb{E}(|V_{\omega}|)\right)\frac{|g||\lambda-\lambda'|(|g-g'|+|\lambda-\lambda'|)}{(|\lambda-\lambda'|^{\delta}+|g-g'|^{\delta})^{3+3d}}.\\
\end{align*}
Using the same arguments as in\cite[Equations 3.17 and 3.18]{H-K-S} we see that 
\begin{equation*}\Big|\frac{(\lambda'-\lambda)}{\pi}\int_{\mathbb{C}}\partial_{\overline{z}}\tilde{\varphi}\,\mathrm{Tr}P_0\varphi(H_{\lambda,g})R_{\lambda,g}(z)V_{\omega}R_{\lambda',g'}(z)P_0\,d^2z\Big|
\end{equation*}
 can be bounded from above by
\begin{equation}
|\lambda-\lambda'||\mathbb{E}\left( \mathrm{Tr}(P_0\varphi(H_{\lambda,g})R_{\lambda,g}(z)V_{\omega}R_{\lambda,g}(z)P_0)\right)|\leq \frac{C|\lambda-\lambda'|\mathbb{E}\left(|V_{\omega}|\right)}{(|\lambda-\lambda'|^{\delta}+|g-g'|^{\delta})}.
\end{equation}
Finally, we conclude that
$|N_{\lambda,g}(E)-N_{\lambda',g'}(E)|$ is bounded from above by
\begin{eqnarray*}C(\alpha_0,d,I)\left(|\lambda-\lambda'|^{\delta \alpha}+|g-g'|^{\delta \alpha}+|\lambda-\lambda'|^{2-(3d+4)\delta}+|g-g'|^{2-(3d+4)\delta}+|\lambda-\lambda'|^{1-\delta}+|g-g'|^{1-\delta}\right).
\end{eqnarray*}
Choosing $\delta=\frac{2}{\alpha+3d+4}$ we obtain, for any $\beta\in [0,\frac{2}{\alpha+3d+4}]$,
\begin{equation*}|N_{\lambda,g}(E)-N_{\lambda',g'}(E)|\leq C(\alpha_0,I)\left(|\lambda-\lambda'|^{\beta}+|g-g'|^{\beta}\right)
\end{equation*}
finishing the proof of theorem \ref{thmids}.

\bibliographystyle{amsplain}

 \end{document}